\documentclass[preprint,12pt]{elsarticle}

\usepackage{amssymb}
\usepackage{amsmath}
\usepackage{graphicx}
\usepackage{amsmath}
\usepackage{latexsym}
\usepackage{amsfonts}
\usepackage{amsthm}
\usepackage{hyperref}
\theoremstyle{definition}
\newtheorem{theorem}{Theorem}[section]
\newtheorem{definition}{Definition}[section]

\usepackage{xargs}
\usepackage[colorinlistoftodos,prependcaption,textsize=tiny]{todonotes}
\newcommandx{\danny}[2][1=]{\todo[linecolor=red,backgroundcolor=red!25,bordercolor=red,#1]{#2}}
\newcommandx{\matt}[2][1=]{\todo[linecolor=blue,backgroundcolor=blue!25,bordercolor=blue,#1]{#2}}
\newcommandx{\joe}[2][1=]{\todo[linecolor=olive,backgroundcolor=olive!25,bordercolor=olive,#1]{#2}}
\newcommandx{\adams}[2][1=]{\todo[linecolor=Plum,backgroundcolor=Plum!25,bordercolor=Plum,#1]{#2}}

\usepackage{xcolor}

\journal{Journal of Computational Physics}

\begin{document}

\begin{frontmatter}

\title{Structure Preservation using Discrete Gradients in the Vlasov-Poisson-Landau System}
\date{\today}
\author[1]{Daniel S. Finn}
\author[2]{Joseph V. Pusztay}
\author[2]{Matthew G. Knepley}
\author[3]{Mark F. Adams}
\affiliation[1]{National Research Council, Naval Research Laboratory Postdoctoral Fellow}
\affiliation[2]{University at Buffalo, The State University of New York}
\affiliation[3]{Lawrence Berkeley National Laboratory}
\begin{abstract}
We present a novel structure-preserving framework for solving the Vlasov-Poisson-Landau system of equations using a particle in cell (PIC) discretization combined with discrete gradient time integrators. The Vlasov-Poisson-Landau system is an accurate model for studying hot plasma dynamics at a kinetic scale where small-angle Coulomb collisions dominate. Our scheme guarantees conservation of mass, momentum and energy as well as preservation of the monotonicity of entropy production in both the time-continuous and discrete systems. We employ the conservative integrator for both the Hamiltonian Vlasov-Poisson equations and the dissipative Landau equation using the PETSc library (\url{www.mcs.anl.gov/petsc}) to showcase structure-preserving properties.
\end{abstract}


\end{frontmatter}
\newpage



\section{Introduction}\label{sec:Introduction}

In recent years, structure-preserving algorithms for simulating kinetic-scale plasmas have been a major field of interest. The goal of these algorithms is to provide a numerical framework for the Vlasov-Maxwell-Landau system, and its simplifications, that preserves the basic laws of physics associated with plasmas. In particular, an ideal numerical framework should conserve mass, momentum and energy, as well as other thermodynamic qualities, like entropy monotonicity, to enable stable long time simulations. Structure-preserving particle-in-cell (PIC) methods~\cite{Morrison2017,Kraus2017_GEMPIC,Qin2016,Shadwick2014,Stamm2014,Xiao2016,Xiao2018}, based on discretizing either the variational or Hamiltonian structure of the underlying kinetic models, have more recently emerged as a primary methodology for simulating plasmas at the kinetic scale not only because they can preserve the energy of the system, but also because they guarantee a local algebraic charge conservation law and the preservation of the symplectic two-form. In practice, however, tradeoffs often exist between the structure-preserving qualities of these algorithms which must be weighed carefully.

If Coulomb collisions are ignored, the Vlasov-Maxwell-Landau system can be simplified into the collisionless Vlasov-Maxwell, and the nonmagnetic, nonrelativistic Vlasov-Poisson, equations. Each of these simplified systems, as well as other related systems~\cite{Chen2011}, have seen significant progress in algorithms that satisfy energy conservation and also preserve other invariants present in the system, such as the momentum and charge conservation, and the divergence-free nature of the magnetic field~\cite{He2015,Squire2012,Xiao2015}. Algorithms that can capture all of these qualities simultaneously are, however, uncommon. Momentum-conserving algorithms are often subject to spurious ``grid heating'' effects~\cite{Werner2025} which break energy conservation while energy-conserving algorithms introduce asymmetries to the discretization that break translational invariance, via Noether's theorem, and thus momentum conservation~\cite{Markidis2011}. The Particle-in-Fourier~\cite{Evstatiev2013,Mitchell2019} approach is a recently developed method that breaks this tradeoff by representing the plasma's charges density directly in Fourier space. As a result, spatial translational symmetry, and thus momentum, is conserved while energy is conserved in the continuous limit and finite-grid instabilities are avoided by the lack of a grid altogether. For this work, we continue development on PIC methods and show how careful choices in discretization can provide long time stable conservation of the mass, momentum and energy, as well as other thermodynamic variables, like entropy.
 
Beyond the advancements in these collisionless, dissipation-free models, there has been little development in obtaining structure-preserving discretizations for the dissipative Coulomb collisions. While binary collision algorithms have existed since at least 1977~\cite{Itoh1988, Nanbu1997}, they have been limited to using equal particle weights or have otherwise lacked energy conservation. These algorithms have furthermore failed to demonstrate algebraic entropy monotonicity. However, construction of a deterministic particle-basis discretization for the homogeneous Landau collision operator has recently been demonstrated by Carrillo et al.~\cite{Carrillo2020}. Although the time-discrete scheme presented did not guarantee moment conservation or entropy monotonicity, a continuous-time structure preservation was achieved. In~\cite{Carrillo2020}, Carrillo et al. also presented the idea that Coulomb collisions can be interpreted as compressible flow driven by an entropy functional enabling a natural introduction of marker particles with arbitrary weights. A discrete-time scheme for simulating the Landau operator while conserving the moments and preserving the entropy monotonicity was presented in~\cite{Hirvijoki2020}. 

In this paper, collisional approaches described in~\cite{Carrillo2020} and~\cite{Hirvijoki2020} are examined in relation to Gonzalez’s discrete gradient framework~\cite{Gonzalez1996}, a methodology for designing numerical integrators that maintain the conservation laws inherent to Hamiltonian systems.
In his work, Gonzalez develops a clear formalism for the design of conservative time-integration schemes for Hamiltonian systems with symmetry. Through the introduction of discrete gradients, implicit second-order conservative schemes can be constructed for general systems which preserve the Hamiltonian along with first integrals that arise from symmetries. Prior studies have explored generalizing the discrete gradients framework for non-canonical Hamiltonian systems, see~\cite{Cohen2011}, often combining the discrete gradient with a suitable discretization of some positive semi-definite or Poisson matrix multiplying the gradient. However, though described in~\cite{Kraus2017,Hirvijoki2018,Eero2020}, implementing a generalized discrete gradient integrator for the Landau operator has remained a more demanding mathematical and computational task. In~\cite{Hu2024}, an implementation of the Landau operator was constructed with discrete gradient integrators, however, the choice in numerical solver limited the computational efficiency of the algorithm. This study presents alternative approaches to both the Hamiltonian and Landau systems that may help overcome the limitations of existing implementations.

Thus, the rest of the paper is structured as follows. We first introduce the Vlasov-Poisson-Landau system of equations, along with some other useful equations for this framework, in Section~\ref{sec:VPLSystem}.
In Section~\ref{sec:BracketFormulation} we introduce the metriplectic formulation for the Vlasov-Poisson-Landau system and review the spatial discretization used for the collisionless Vlasov-Poisson system and  the collisional Landau equation. In Section~\ref{sec:TemporalDiscretization}, we then present the temporal discretization for both collisionless and collisional systems, starting with a review of the discrete gradient methodology. For the Landau operator, we present a marker-particle discretization, as first shown in~\cite{Carrillo2020,Hirvijoki2021}, and an altered form of the discrete-time integrator that uses discrete gradients to retain structure-preserving properties of the continuous system. We break from Gonzalez's formalism in our choice of discrete gradients as his specific choice can encounter singularities near equilibrium plasma states~\cite{Liu2023}. Throughout Sections~\ref{sec:BracketFormulation} and~\ref{sec:TemporalDiscretization}, we focus on the impact the discretization has on the structure-preserving properties of the continuous system. 
In Section~\ref{sec:NumImp}, we describe a new numerical implementation of both the original Hamiltonian discrete gradient integrator and the new time-discrete integrator, based in the Portable Extensible Toolkit for Scientific Computing (\texttt{PETSc})~\cite{AbhyankarBrownConstantinescuGhoshSmith2014,petsc-web-page,PETScManual}. For the collisionless operator, we test our new implementation with the classic plasma test, Landau damping, focusing on how well the discrete gradient integrator preserves the structure of the plasma system and how various nonlinear solvers impact the accuracy implicit integrator. In the case of the Landau operator, we use a two-species thermalization test to show structure preservation and entropy monotonicity using the discrete gradients based collisional integrator. In both tests, we also compare the discrete gradients method to other time integrators used in previous work to show the advantages of discrete gradients. 

\section{The Vlasov-Poisson-Landau System}\label{sec:VPLSystem}

We consider the Vlasov-Poisson-Landau system of equations, describing the evolution of the mass distribution function $f_s (\boldsymbol{z},t)$ of species $s$ in phase space $\boldsymbol{z} = (\boldsymbol{x}, \boldsymbol{v}) \in \Omega \times \mathbb{R}^{d}$, given by,
\begin{align}
  \frac{d f_{s}}{d t} \equiv \frac{\partial f_{s}}{\partial t}+\boldsymbol{v} \cdot \nabla_{\boldsymbol{x}} f_{s}+\frac{q_s}{m_s} \boldsymbol{E} \cdot \nabla_{\boldsymbol{v}} f_{s} &= C_{s \bar{s}} (f_s, f_{\bar{s}}),\label{eq:VML}\\
  -\epsilon_0 \Delta \phi &= \rho,\label{eq:Poisson}\\
  \boldsymbol{E} &= -\nabla \phi\label{eq:electrostatic},
\end{align}
where $q_s$ and $m_s$ are the species' charge and mass, respectively, $\phi$ the electrostatic potential, $\rho$ the charge density, and $\boldsymbol{E}$ the electric field. The Vlasov operator $d/dt$  describes the streaming of particles influenced by the electrostatic forces~\cite{Vlasov1938}, and the Landau collision operator $C_{s \overline{s}}(f_s, f_{\bar{s}})$ describes the effects due to small-angle Coulomb collisions between the species s and $\bar{s}$~\cite{Landau1936}. We define the Landau collision operator as,
\begin{align}
C_{s \overline{s}}(f_s, f_{\bar{s}}):= \sum_{\bar{s}} -\frac{\nu_{s \bar{s}}}{m_s} \nabla_{\boldsymbol{v}} \cdot \left[ \int \delta (\boldsymbol{x} - \bar{\boldsymbol{x}}) f_s(\boldsymbol{z}) f_{\bar{s}}(\bar{\boldsymbol{z}}) \mathbb{Q}\left(\boldsymbol{v}-\bar{\boldsymbol{v}}\right) \cdot \boldsymbol{\Gamma}_{s \bar{s}} (\mathcal{S}, \boldsymbol{z},\bar{\boldsymbol{z}}) \mathrm{d} \boldsymbol{\bar{z}}\right].
\label{eq:ColOp}
\end{align}
In this formulation, $\nu_{s \bar{s}} = 2\pi q_s q_{\bar{s}} \ln \Lambda$, $\ln \Lambda$ is the Coulomb logarithm, and the vector $\boldsymbol{\Gamma}_{s \bar{s}} (\mathcal{F}, \boldsymbol{z},\bar{\boldsymbol{z}},t)$ is defined by the equation,
\begin{equation}
  \boldsymbol{\Gamma}_{s \bar{s}}(\mathcal{F}, \boldsymbol{z}, \overline{\boldsymbol{z}}):=\frac{1}{m_{s}} \frac{\partial}{\partial \boldsymbol{v}} \frac{\delta \mathcal{F}}{\delta f_{s}}(\boldsymbol{z})-\frac{1}{m_{\bar{s}}} \frac{\partial}{\partial \overline{\boldsymbol{v}}} \frac{\delta \mathcal{F}}{\delta f_{\bar{s}}}(\overline{\boldsymbol{z}}),
\end{equation}
for some arbitrary functional $\mathcal{F}$. In addition, $\mathbb{Q} (\boldsymbol{\xi})$ is the collision kernel, 
\begin{equation}
  \mathbb{Q} (\boldsymbol{\xi}) = \frac{1}{|\boldsymbol{\xi}|} \left( \mathbb{I} - \frac{\boldsymbol{\xi} \boldsymbol{\xi}}{|\boldsymbol{\xi}|^2} \right),
\end{equation}
with $\mathbb{I}$ being the identity matrix. It is worthwhile to note here that $\mathbb{Q}(\boldsymbol{\xi})$ is also a scaled projection onto the plane perpendicular to $\boldsymbol{\xi}$. This property of $\mathbb{Q} (\boldsymbol{\xi})$ will be important in later sections. We also define several useful values for later sections. The charge density, $\rho$, is obtained from the distribution function using the relation,
\begin{equation}
    \rho = \sum_s q_s \int f_s  \mathrm{d} \mathbf{z}.\label{eq:chargedensity}
\end{equation}
In single species models, slowly moving ions may be replaced with a constant neutralizing background charge density, $\sigma$, subtracted from the calculated charge density, $\rho$, at each step. The Hamiltonian of the system is given by the sum of the kinetic energy and the electric field energy, namely,
\begin{equation}\label{eq:Hamiltonian}
  \mathcal{H} = \sum_s \frac{m_s}{2} \int f_s \left|\boldsymbol{v}\right|^2 \mathrm{d} \boldsymbol{z} + \frac{1}{2} \int \left| \nabla \phi \right|^2 \mathrm{d} \boldsymbol{z}.
\end{equation}
The primary moments of the plasma system (mass, momentum and kinetic energy), as well as the system entropy, are defined in their continuous forms by,
\begin{align}
  \mathcal{M} &= \sum_s \int  m_s  f_s \mathrm{d} \boldsymbol{z}, \label{eq:MassFunc} \\
  \mathcal{P} &= \sum_s \int  m_s |\boldsymbol{v}| f_s \mathrm{d} \boldsymbol{z}, \label{eq:MomFunc} \\
  \mathcal{K} &= \sum_s \int \frac{1}{2} m_s |\boldsymbol{v}|^2 f_s \mathrm{d} \boldsymbol{z}, \label{eq:KinEnFunc} \\
  \mathcal{S} &= - \sum_s \int f_s \ln f_s \mathrm{d} \boldsymbol{z}. \label{eq:EntFunc}
\end{align}
Lastly, we define the functional derivative using the Fr\'echet derivative,
\begin{equation}
  \left.\frac{\partial}{\partial \epsilon}\right|_{\epsilon=0} \mathcal{F}\left[f_{s}+\epsilon \delta f_{s}\right]=\int \frac{\delta \mathcal{F}}{\delta f_{s}} \delta f_{s} \mathrm{d} \boldsymbol{z} \equiv \delta \mathcal{F}\left[\delta f_{s}\right],
  \label{FuncDer}
\end{equation}
where $\mathcal{F}$ is any arbitrary functional.

We note that the following derivation applies to the magnetic Vlasov-Maxwell-Landau system of equations as well. However, our current implementation considers the electrostatic limit in which speed of light, $c$, goes to infinity ($c\rightarrow \infty$). This may be viewed as the case where the plasma frequency is much slower than the frequency of the electromagnetic waves, thus, we will for now consider the Vlasov-Poisson-Landau model~\cite{Miloshevich2021}.

\section{Bracket Formulation}\label{sec:BracketFormulation}

The Vlasov-Poisson-Landau equations~\eqref{eq:VML}-\eqref{eq:electrostatic} display both Hamiltonian and dissipative dynamics. In order to handle this, we introduce the concept of metriplectic dynamics~\cite{Morrison1986}. The dissipation-free Vlasov-Poisson system has an infinite-dimensional non-canonical Hamiltonian structure culminating in a infinite-dimensional Poisson bracket $\{\cdot ,~\cdot \}$ with a Hamiltonian, $\mathcal{H}$, as the generating function of the flow. The collisional Landau system can be formulated as an infinite-dimensional metric bracket $( \cdot , \cdot)$ generated by the entropy functional $\mathcal{S}$. We can then define a combined system, for which the evolution of an arbitrary functional $\mathcal{F}$ in time is given by,
\begin{equation}
\frac{d \mathcal{F}}{dt} = \{ \mathcal{F}, \mathcal{H}\} + \left( \mathcal{F}, \mathcal{S} \right). \label{eq:MetriplecticBracket}
\end{equation} 

Before proceeding with any further discussion of the two brackets, we must take a step back and address a few questions that will be raised in the next steps. For both the Poisson and metric brackets we start with the common marker-particle representation for the distributional density,
\begin{equation}
    f_h (\boldsymbol{z}, t) \mathrm{d} \boldsymbol{z} = \sum_p w_p \delta (\boldsymbol{x} - \boldsymbol{\boldsymbol{x}_p} (t)) \delta (\boldsymbol{v} - \boldsymbol{\boldsymbol{v}_p} (t)) \mathrm{d} \boldsymbol{z},
    \label{MarkPartDist}
\end{equation}
where $\boldsymbol{x}_p(t)$, $\boldsymbol{v}_p(t)$ and $w_p$, are the position, velocity and weight of particle $p$, respectively. We first encounter a question when considering the evaluation of the functional derivatives~\eqref{FuncDer} with respect to the distribution~\eqref{MarkPartDist}. This is a essential piece of structure-preserving particle-based discretizations tackled in numerous previous derivations~\cite{Morrison1981,Kraus2017_GEMPIC,Carrillo2020}. We will refer to the solution in~\cite{Eero2020} and give a brief summary here. If the distribution function $f_s$ is restricted to the specific parametrized form $f_h$ given in~\eqref{MarkPartDist}, then the perturbation of $f_s$ also becomes parametrized and takes the form,
\begin{align}
\delta f_{h} \mathrm{d} \boldsymbol{z}=-\sum_{p} w_{p}&\left(\delta\left(\boldsymbol{v}-\boldsymbol{v}_{p}\right) \nabla_{\boldsymbol{x}_p} \delta\left(\boldsymbol{x}-\boldsymbol{x}_{p}\right) \cdot \delta \boldsymbol{x}_{p}\right.\\
&\left.+\delta\left(\boldsymbol{x}-\boldsymbol{x}_{p}\right) \nabla_{\boldsymbol{v}_p} \delta\left(\boldsymbol{v}-\boldsymbol{v}_{p}\right) \cdot \delta \boldsymbol{v}_{p}\right) \mathrm{d} \boldsymbol{z}.
\end{align}
Accordingly, the variation of the functional $\mathcal{F} \left[ f \right]$, being restricted to distributions of type~\eqref{MarkPartDist}, then becomes,
\begin{equation}
\delta \mathcal{F}\left[\delta f_{h}\right]=\sum_{p} w_{p}\left(\left.\nabla_{\boldsymbol{x}_p}\frac{\delta \mathcal{F}}{\delta f_h}\right|_{\boldsymbol{z}_{p}} \cdot \delta \boldsymbol{x}_{p}+\left.\nabla_{\boldsymbol{v}_p} \frac{\delta \mathcal{F}}{\delta f_h}\right|_{\boldsymbol{z}_{p}} \cdot \delta \boldsymbol{v}_{p}\right).
\end{equation} 
When evaluated with respect to the distribution $f_h$, the functional $\mathcal{F}$, however, becomes an ordinary function of the degrees of freedom. That is to say,
\begin{equation}
\mathcal{F} \left[ f_h \right] = F (\boldsymbol{X}, \boldsymbol{V}; W),
\end{equation}
where $\boldsymbol{X} = \{ \boldsymbol{x}_p \}_p$ and $\boldsymbol{V} = \{ \boldsymbol{v}_p \}_p$ denote the collection of marker-particle degrees of freedom, and $W = \{ w_p \}_p$ denotes the collection of particles fixed weights. Therefore, $\delta \mathcal{F} \left[ \delta f_h \right]$ must be equal to the variation of the function $F \left(\boldsymbol{X}, \boldsymbol{V}; W\right)$. This leads to the identities,
\begin{align}
\left.\nabla_{\boldsymbol{x}_p} \frac{\delta \mathcal{F}}{\delta f}\right|_{\boldsymbol{z}_{p}} &=\frac{1}{w_{p}} \nabla_{\boldsymbol{x}_p} F(\boldsymbol{X}, \boldsymbol{V} ; W) \label{Identity1}\\
\left.\nabla_{\boldsymbol{v}_p} \frac{\delta \mathcal{F}}{\delta f}\right|_{\boldsymbol{z}_{p}} &=\frac{1}{w_{p}} \nabla_{\boldsymbol{v}_p} F(\boldsymbol{X}, \boldsymbol{V} ; W). \label{Identity2}
\end{align}

Secondly, if the entropy functional~\eqref{eq:EntFunc} is evaluated with respect to the marker-particle distribution~\eqref{MarkPartDist}, it is not well defined. In particular, as will be shown in Section~\ref{sec:DGDI}, we require an entropy functional for which a gradient can be taken. Furthermore, the entropy functional only gives information from individual particles without considering how particle-particle interactions may affect the entropy of the system. If the global mass, or charge, of the system is conserved, then~\eqref{eq:EntFunc} will be conserved as well and we learn little about our system. To resolve these issues, Carrillo et al.~\cite{Carrillo2020}, proposed introducing a regularized entropy functional where the distribution function~\eqref{MarkPartDist} is first convolved with a radial basis function, also called a \textit{mollifier}, $\psi_{\epsilon}$. In practice, the mollifier can be any radial basis function, but we will focus on the case where the mollifier is a Gaussian, 
\begin{equation}
\psi_{\epsilon}(\boldsymbol{v})=\frac{1}{(2 \pi \epsilon)^{\frac{d}{2}}} \exp \left(-\frac{|\boldsymbol{v}|^{2}}{2 \epsilon}\right),
\end{equation}
for any $\epsilon > 0$. The standard entropy functional can then be replaced by the regularized entropy functional,
\begin{equation}
\mathcal{S} \left[ f_s \right] \Rightarrow \mathcal{S} \left[ \psi_{\epsilon} * f_s \right] = \mathcal{S}_\epsilon \left[ f_s \right], 
\label{RegEntSwitch}
\end{equation}
which enables the evaluation of the entropy functional even for the delta distribution~\eqref{MarkPartDist}. Furthermore, the delta functions appearing in the entropy functional \eqref{RegEntSwitch} can simply be replaced with radial basis functions $\psi_{\epsilon}$ so that,
\begin{equation}
\mathcal{S}_\epsilon \left[ f_s \right]=-\sum_s \int \left(f_s \ast \psi_{\epsilon}\right) \ln \left( f_s \ast \psi_{\epsilon} \right) d\boldsymbol{z}.
\label{RegEntropy}
\end{equation}

\subsection{Poisson Bracket}\label{sec:PoissonBracket}
A non-canonical Poisson bracket for the Vlasov-Maxwell system was introduced in~\cite{Morrison1980}, with term corrections in~\cite{Marsden1982} and its limitation to divergence-free magnetic fields first pointed out in~\cite{Morrison1982}. If we assume the magnetic-free Vlasov-Poisson limit, we may write the Poisson bracket for arbitrary functionals $\mathcal{F}$ and $\mathcal{G}$ as,
\begin{align}\label{eq:PoissonBracket}
  \left\{\mathcal{F},\mathcal{G}\right\} \left[f_s\right] = &\sum_s \int \frac{f_s}{m_s} \left[\frac{\delta \mathcal{F}}{\delta f_s},\frac{\delta \mathcal{G}}{\delta f_s}\right] d\boldsymbol{z}
\end{align}
where $\left[f,g\right] = \nabla_{\boldsymbol{x}} f \cdot \nabla_{\boldsymbol{v}} g - \nabla_{\boldsymbol{x}} g \cdot \nabla_{\boldsymbol{v}} f$ and all the field terms have vanished from the full Vlasov-Maxwell bracket. With the Hamiltonian functional~\eqref{eq:Hamiltonian}, we can represent the time evolution of the functional, $\mathcal{F}$, as,
\begin{equation}
  \frac{d \mathcal{F}}{dt} = \left\{\mathcal{F},\mathcal{H}\right\}.
  \label{eq:PoissonBracket1}
\end{equation}

If the functional $\mathcal{F}$ is given by the continuous distributional density
\begin{equation}
    \mathcal{F}[f] = \int f(\bar{\boldsymbol{x}},\bar{\boldsymbol{v}},t)\, \delta(\boldsymbol{x}-\bar{\boldsymbol{x}})\,\delta(\boldsymbol{v}-\bar{\boldsymbol{v}})\,d\bar{\boldsymbol{z}},
\end{equation}
then the Poisson bracket~\eqref{eq:PoissonBracket1} recovers the collisionless Vlasov–Poisson equation (Eq.~\eqref{eq:VML} with $C_{s\bar{s}}=0$) via the method of characteristics. 
In this formulation, the Eulerian distribution is transported by the Hamiltonian phase-space flow 
\(\Phi_t\), such that \(f(\boldsymbol{z},t)=f_0(\Phi_t^{-1}(\boldsymbol{z}))\). 
The corresponding characteristic equations describe the Lagrangian trajectories in phase space, which form the foundation for particle-based discretizations of the system. 
In practical particle formulations, these trajectories are obtained by discretizing the continuum Poisson bracket using a marker-particle ansatz, as discussed later in this section. 
This finite-dimensional reduction yields canonical equations for the particle coordinates $\boldsymbol{z}_p = (\boldsymbol{x}_p, \boldsymbol{v}_p)$, providing a discrete realization of the same Hamiltonian structure governing the continuum Vlasov–Poisson dynamics.

A family of conserved quantities, called \textit{Casimir Invariants} (or \textit{Casimirs}), are defined as any functional that commutes with any other arbitrary functional, i.e. a functional $\mathcal{C}$ is invariant if $\left\{\mathcal{C},\mathcal{G}\right\}=0$ for any functional, $\mathcal{G}$. In the context of this problem, any functional of the form $\mathcal{C}\left[f_h\right]=\int C(f_s) \mathrm{d} \mathbf{z}$ is a Casimir, which includes the mass~\eqref{eq:MassFunc} and entropy~\eqref{eq:EntFunc}. 
The momentum~\eqref{eq:MomFunc}, and regularized entropy functionals~\eqref{RegEntropy} may be shown to commute with the Hamiltonian operator, proving these are conserved quantities but not necessarily Casimirs. 
Furthermore, it can be shown that the Hamiltonian~\eqref{eq:Hamiltonian} (or total energy functional) commutes with itself, $\left\{\mathcal{H},\mathcal{H}\right\}=0$, as the bracket is antisymmetric ensuring it is conserved in the continuous, collisionless Vlasov-Poisson form. Proofs of the conservation of the moments, Hamilton functional, and both forms of entropy are shown in~\ref{sec:AppendixA}.

\subsection{Finite Element Discretization of the Poisson Bracket}\label{sec:FEPoissonBracket}
The goal of this spatial and temporal discretization for the Poisson bracket is to preserve as much structure as possible from the continuous form. In our case, we hope to preserve the antisymmetry of the bracket, conserved quantities, and Casimir invariants. To discretize the fields, $\phi$ and $\boldsymbol{E}$, we may introduce compatible finite element spaces for the differential forms:
\begin{itemize}
    \item \textbf{0-forms (scalars)}: A finite element space $V^0_h \subset H^{1} (\Omega)$ used for the electrostatic potential $\phi_h$ with basis functions $\left\{\psi_i\right\}_{i=1}^{N_0}$ spanning $V^0_h$.
    \item \textbf{1-forms (vector fields)}: A finite element space $V^1_h  \subset H (curl;\Omega)$ used for the electric field $\boldsymbol{E}_h$ with basis 1-forms $\left\{\varphi_j\right\}_{j=1}^{N_1}$ spanning $V^1_h$.
\end{itemize}
\noindent With discrete spaces for both fields, $\phi$ and $\boldsymbol{E}$, selected we can proceed with mixed finite element formulation, using Finite Element Exterior Calculus (FEEC) in the style of~\cite{Kraus2017_GEMPIC}. Mixed formulations are locally conservative cell-by-cell, to solver tolerance, as the discrete gradient, curl, and divergence operators commute with interpolation. However, these formulations require increased solver complexity and memory usage which in this current implementation is already high. Therefore, to simplify our implementation, we use the \textit{primal} form with only one field, $\phi$, and proceed with derivation ignoring the $V^1_h$ space.

From the definition of our primal discrete space, we can redefine the field terms in their discrete forms. We start by introducing a mass matrix,
\begin{equation}
    \mathbb{M}_{ij}=\int \psi_i(x) \psi_j(x) dx,
\end{equation}
which represents the $L^2$ inner product on the 0-form. The deposition of charge onto the 0-form basis is calculated by,
\begin{align}
    (\rho_h)_i &= q \int f_h \psi_i (\mathbf{x}) \mathrm{d} \boldsymbol{z}, \\
    &= \int \sum_p q_p w_p \delta (\boldsymbol{x} - \boldsymbol{x}_p(t)) \delta (\boldsymbol{v} - \boldsymbol{v}_p(t)) \psi_i (\boldsymbol{x}) \mathrm{d} \boldsymbol{z}, \\
    &= \sum_p q_p w_p \psi_i (\mathbf{x}_p).
\end{align}
If the potential, $\phi$, is discretized using the basis functions $\psi_i$,
\begin{equation}
    \phi_h (\mathbf{x}) = \sum_{i} \phi_i \psi_i (\mathbf{x}),
\end{equation}
then the discretized electric field can be calculated as,
\begin{align}\label{eq:discEfield}
    \boldsymbol{E}_h(x) &= -\nabla \phi_h (\mathbf{x}), \\
    &= - \sum_{i} \phi_i \nabla \psi_i (\mathbf{x}).
\end{align}

Combining the particle coordinates and field coefficients into a single state vector, let $\mathbf{X}=(\boldsymbol{x}_1,\dots,\boldsymbol{x}_{N_p})$ and $\mathbf{V}=(\boldsymbol{v}_1,\dots,\boldsymbol{v}_{N_p})$ be the stacked particle positions and velocities. The finite‐dimensional Poisson bracket is generated by a skew-symmetric matrix, $\mathbb{J}$, so that,
\begin{equation}
  \left\{F,G\right\} = \left(\nabla_{\boldsymbol{u}} F\right)^{T} \mathbb{J}(\boldsymbol{u}) \nabla_{\boldsymbol{u}} G,\label{eq:discPoissonBracket}
\end{equation}
where $\boldsymbol{u}=\left(\boldsymbol{X},\boldsymbol{V}\right)$. In 3D, each particle $p$ contributes a $6\times6$ canonical Poisson block. This gives, \begin{align}
        \mathbb{M}_p &= diag\left(w_1 m_1,\ldots,w_{N_p} m_{N_p}\right) \otimes \mathbb{I}_3,\text{ and} \\
        \mathbb{M}_p^{-1} &= diag \left(\frac{1}{w_1 m_1},\ldots,\frac{1}{w_{N_p} m_{N_p}}\right) \otimes \mathbb{I}_3,
    \end{align}
representing the full configuration space.
The canonical coupling between particle positions, $\mathbf{X}$, and velocities, $\mathbf{V}$, in the Poisson matrix is then, 
\begin{align}
    \mathbb{J} &= \begin{pmatrix}
        0 & \mathbb{M}_p^{-1} \\
        (\mathbb{M}_p^{-1})^T & 0 \\
    \end{pmatrix}.
\end{align}
In the Vlasov-Maxwell case, there are additional field and particle-field coupling blocks, however in the electrostatic limit these go to zero. 
Writing out the Poisson bracket gives, 
\begin{align}
    \left\{F,G\right\} = &\sum_p \frac{1}{m_p w_p}\left(\nabla_{\boldsymbol{x}_p} F \cdot \nabla_{\boldsymbol{v}_p} G - \nabla_{\boldsymbol{x}_p} G \cdot \nabla_{\boldsymbol{v}_p} F \right).\label{eq:DiscPoissonBracketPointwise}
\end{align}
In this spatially discrete form, the conservation properties are easier to show for each of the moments, as well as the two forms of entropy. We first redefine the moment functionals as ordinary functions of the degrees of freedom,
\begin{gather}
M(\boldsymbol{X}, \boldsymbol{V} ; W)=\sum_{p} w_{p} m_p, \label{eq:DiscMass}\\
P(\boldsymbol{X}, \boldsymbol{V} ; W)=\sum_{p} w_{p} m_p \boldsymbol{v}_{p}, \label{eq:DiscMomentum}\\
K(\boldsymbol{X}, \boldsymbol{V} ; W)=\frac{1}{2} \sum_{p} w_{p} m_p \left|\boldsymbol{v}_{p}\right|^{2}\label{eq:DiscKinEn}.
\end{gather}
Furthermore, in this form the entropy functionals are redefined as,
\begin{gather}
S(\boldsymbol{X}, \boldsymbol{V} ; W)=-\sum_{p} w_{p} \ln w_p, \label{eq:DiscEntropy}\\
S_{\epsilon}(\boldsymbol{X}, \boldsymbol{V} ; W)=-\int \sum_{p} w_{p} \psi_{\epsilon}\left(\boldsymbol{z}-\boldsymbol{z}_{p}\right) \ln \left(\sum_{p^{\prime}} w_{p^{\prime}} \psi_{\epsilon}\left(\boldsymbol{z}-\boldsymbol{z}_{p^{\prime}}\right)\right) \mathrm{d} \boldsymbol{z},\label{eq:DiscRegEntropy}
\end{gather}
and the Hamiltonian function is now,
\begin{equation}\label{eq:DiscHamiltonian}
  H = \frac{1}{2} \sum_p w_p m_p \left|\boldsymbol{\boldsymbol{v}_p}\right|^2 + \frac{1}{2} \phi^T L \phi.
\end{equation}
where $L = \int \nabla \psi_i (\boldsymbol{x}) \cdot \nabla \psi_j (\boldsymbol{x}) \mathrm{d} \boldsymbol{x}$ is the discrete Laplace (or stiffness) matrix.

With the finite-dimensional Poisson bracket and discrete Hamiltonian, $H$, are introduced, the resulting structure naturally yields the canonical equations of motion for each particle, 
\begin{align}
  \frac{d\boldsymbol{x}_p}{dt} &= \{\boldsymbol{x}_p,H\} = \boldsymbol{v}_p, \\
  \frac{d\boldsymbol{v}_p}{dt} &= \{\boldsymbol{v}_p, H\} = -\frac{q_p}{m_p w_p} \boldsymbol{E}(\boldsymbol{x}_p),
\end{align}
thus recovering the standard single-particle dynamics as the finite-dimensional reduction of the continuum Vlasov–Poisson system.
An early example of this explicit reduction was presented by Morrison in~\cite{Morrison1981}.

The mass and entropy functions,~\eqref{eq:DiscMass} and~\eqref{eq:DiscEntropy}, don't depend on either of the dynamical variables $\mathbf{X}$ and $\mathbf{V}$, so the gradients all vanish,
\begin{align*}
    \nabla_x M &= \nabla_v M = 0,\\
    \nabla_x S &= \nabla_v S = 0,
\end{align*}
and thus $\left\{ M,G \right\} = 0$ and $\left\{ S,G \right\} = 0$, which are both Casimirs of the spatially discrete bracket as well. As with the continuous bracket, we prove that the momentum function~\eqref{eq:DiscMomentum} is a conserved quantity and not a Casimir by considering the Poisson bracket with the momentum and Hamiltonian functions,
\begin{align*}
    \left\{P,H\right\} &= \sum_p \frac{1}{m_p w_p}\left(\nabla_{\boldsymbol{x}_p} P \cdot \nabla_{\boldsymbol{v}_p} H - \nabla_{\boldsymbol{x}_p} H \cdot \nabla_{\boldsymbol{v}_p} P \right) \\
    &= \sum_p \frac{1}{m_p w_p}\left( \mathbf{0} \cdot \nabla_{\boldsymbol{v}_p} H - \nabla_{\boldsymbol{x}_p} H \cdot (w_p m_p \mathbf{1}) \right) \\
    &= - \sum_p \nabla_{\boldsymbol{x}_p} H.
\end{align*}
We recall that the discrete Hamiltonian depends on the particle positions only through the self-consistent electrostatic potential, which is obtained from a finite element discretization of Poisson's equation. In particular, given the discrete electric field~\eqref{eq:discEfield}, the total force on all particles can be written as,
\begin{align*}
    \sum_p \boldsymbol{F}_p &= - \sum_p \nabla_{\boldsymbol{x}_p} H,\\
    &= \sum_p q_p \sum_{i} \phi_i \nabla \psi_i (\boldsymbol{x}_p), \\
    &= \sum_{i} \phi_i \sum_p q_p \nabla \psi_i (\boldsymbol{x}_p).
\end{align*}
For the class of finite element discretizations considered here, whose basis functions form a partition of unity and provide an exact representation of the charge density, the electrostatic forces between particles cancel exactly when summed over the system. As a result, the discrete Hamiltonian produces no net self-force, i.e. $\sum_p \boldsymbol{F}_p = - \sum_p \nabla_{\boldsymbol{x}_p} H = 0$. This cancellation relies on the translational invariance of the discrete Hamiltonian, which is ensured here by the use of periodic boundary conditions. In the absence of periodicity, boundary forces break translation invariance and the total force need not vanish. Therefore, $\left\{P,H\right\} = \mathbf{0}$ and momentum remains a conserved quantity of the discretized bracket. These properties are satisfied, for example, by the compatible finite element discretizations employed in the GEMPIC framework~\cite{Kraus2017_GEMPIC}.

From~\eqref{eq:DiscPoissonBracketPointwise}, we see that the antisymmetry of the continuous Poisson bracket is preserved in the discrete form and total energy~\eqref{eq:DiscHamiltonian} remains conserved in the discrete Poisson bracket, as well. For regularized entropy~\eqref{eq:DiscRegEntropy}, we first consider how the Poisson bracket reduces with the Hamiltonian function,
\begin{align*}
    \left\{S_{\epsilon},H\right\} &= \sum_p \nabla_{\boldsymbol{x}_p} S_{\epsilon} \cdot \boldsymbol{v}_p.
\end{align*}
Now, consider the gradient of~\eqref{eq:DiscRegEntropy} with respect to $\boldsymbol{x}_p$,
\begin{align*}
    \nabla_{\boldsymbol{x}_p} S_{\epsilon} &= w_p \int \nabla_{\boldsymbol{x}_p} \psi_{\epsilon} (\boldsymbol{z}-\boldsymbol{z}_p) \left(\log h(\boldsymbol{z})+1\right) \mathrm{d} \boldsymbol{z},
\end{align*}
where $h(\boldsymbol{z}) = \sum_{p'} w_{p'} \psi_{\epsilon} (\boldsymbol{z} - \boldsymbol{z}_{p'})$. The bracket may be rewritten using integration by parts as,
\begin{align*}
     \left\{S_{\epsilon},H\right\} &= - \sum_p w_p \boldsymbol{v}_p \cdot \int \psi_{\epsilon} (\boldsymbol{z}-\boldsymbol{z}_p) \nabla_{\boldsymbol{x}_p} \left(\log h(\boldsymbol{z})+1\right) \mathrm{d} \boldsymbol{z},\\
     &= - \sum_p w_p \boldsymbol{v}_p \cdot \int \psi_{\epsilon} (\boldsymbol{z}-\boldsymbol{z}_p)  \frac{\nabla_{\boldsymbol{x}_p} h (\boldsymbol{z})}{h (\boldsymbol{z})} \mathrm{d} \boldsymbol{z}.
\end{align*}
It can be shown that the integrand is antisymmetric in the particle indices ($p$ and $p'$) and since the domain is symmetric all terms in the bracket cancel. We then note that for any smooth, rapidly decaying function kernel, $\psi_{\epsilon}$,
\begin{align*}
    \int \nabla_v \psi_{\epsilon} (\boldsymbol{z}-\boldsymbol{z}_p) f(\boldsymbol{z}) \mathrm{d} \boldsymbol{z} = 0,
\end{align*}
when boundaries are periodic and the regularized entropy remains conserved. We have, therefore, discretized the Poisson bracket in space and shown that each of the thermodynamic quantities preserved in the original continuous bracket have retained their same structure in this spatially discrete form. From here, consider the metric bracket for the collisional Landau operator and the the temporal discretization of both brackets with discrete gradients.

\subsection{Metric Bracket}\label{sec:MetricBracket}

The metric bracket formulation for the Landau collision operator was first introduced in~\cite{Morrison1984} and~\cite{Kaufman1984}. In~\cite{Morrison1986}, the full formalism of the metric bracket was then defined, along with the properties such as symmetry and degeneracy of the bracket. The following derivation follows from these works, as well as the discretizations provided in~\cite{Kraus2017},~\cite{Carrillo2020}, and~\cite{Eero2020}. 

The first step of the derivation of the metric bracket formulation is to obtain the weak form of the collision operator~\eqref{eq:ColOp}. We do this by multiplying~\eqref{eq:ColOp} by a test function, which in this case is an arbitrary species-dependent function $G_s (\boldsymbol{z})$, and integrating over the configuration space,
\begin{align}
\int G_s (\boldsymbol{z}) \frac{\partial f_s}{\partial t} \mathrm{d} \boldsymbol{z} &= \sum_{s} \int G_s (\boldsymbol{z}) C_{s \overline{s}}(f_s, f_{\bar{s}}) \mathrm{d} \boldsymbol{z}, \nonumber\\
&=\sum_{s, \bar{s}} \iint \frac{1}{m_{s}} \frac{\partial G_{s}(\boldsymbol{z})}{\partial \boldsymbol{v}} \cdot \mathbb{W}_{s \bar{s}}(\boldsymbol{z}, \overline{\boldsymbol{z}}) \cdot \boldsymbol{\Gamma}_{s \bar{s}}(\mathcal{S}, \boldsymbol{z}, \overline{\boldsymbol{z}}) d \overline{\boldsymbol{z}} \mathrm{d} \boldsymbol{z}, \nonumber \\
&=\sum_{s, \bar{s}} \frac{1}{2} \iint\left(\frac{1}{m_{s}} \frac{\partial G_s (\boldsymbol{z})}{\partial \boldsymbol{v}}-\frac{1}{m_{\bar{s}}} \frac{\partial G_{\bar{s}}(\overline{\boldsymbol{z}})}{\partial \overline{\boldsymbol{v}}}\right)  \label{WeakForm}\\
&\;\;\;\cdot \mathbb{W}_{s \bar{s}}(\boldsymbol{z}, \overline{\boldsymbol{z}}) \cdot \boldsymbol{\Gamma}_{s \bar{s}}(\mathcal{S}, \boldsymbol{z}, \overline{\boldsymbol{z}}) d \overline{\boldsymbol{z}} \mathrm{d} \boldsymbol{z},\nonumber
\end{align}
where the positive semidefinite matrix $\mathbb{W}_{s \bar{s}}(\boldsymbol{z}, \overline{\boldsymbol{z}})$ is defined as,
\begin{equation}
  \mathbb{W}_{s \bar{s}}(\boldsymbol{z}, \overline{\boldsymbol{z}}) = \nu_{s \bar{s}} \delta(\boldsymbol{x}-\overline{\boldsymbol{x}}) f_{s}(\boldsymbol{z}) f_{\bar{s}}(\overline{\boldsymbol{z}}) \mathbb{Q}(\boldsymbol{v}-\overline{\boldsymbol{v}}).
  \label{Wmat}
\end{equation}

In this subsection, we reduce the configuration space, $\boldsymbol{z}=(\boldsymbol{x},\boldsymbol{v})$, to the spatially homogeneous case, $\boldsymbol{z}=(\boldsymbol{v})$, as is done in~\cite{Carrillo2020,Eero2020}, as the Landau collision operator~\eqref{eq:ColOp} is dependent on velocity alone. Since the electric potential, $\phi$, is a function of space and not velocity, this has the effect of reducing the full Hamiltonian functional to just the kinetic energy functional. Given the weak form~\eqref{WeakForm}, it is now possible to prove conservation of the moments for the continuous Landau operator using the mass, momentum and kinetic energy functionals. Since the left hand side of~\eqref{WeakForm}, with $G_s (\boldsymbol{z})$ replaced by these moment functionals, presents the total collisional rate of change in time of mass density, momentum density and kinetic energy density in the plasma, we can show that the three moments are invariants of the Landau collision operator, and thus conserved, by showing that the right hand side of~\eqref{WeakForm} goes to zero in each case. The conservation of mass is trivial as~\eqref{eq:MassFunc} is not a function of $\boldsymbol{v}$, thus it goes to zero when differentiated in~\eqref{WeakForm}. Similarly, conservation of momentum is simple to show as,
\begin{align*}
\left(\frac{1}{m_{s}} \frac{\partial \mathcal{P}_s (\boldsymbol{z})}{\partial \boldsymbol{v}}-\frac{1}{m_{\bar{s}}} \frac{\partial \mathcal{P}_{\bar{s}}(\overline{\boldsymbol{z}})}{\partial \overline{\boldsymbol{v}}}\right) &= \mathbb{I} - \mathbb{I}, \\
&= 0,
\end{align*}
therefore the right hand side of~\eqref{WeakForm} goes to zero as well. Lastly, conservation of kinetic energy arises from the projection nature of $\mathbb{Q}(\boldsymbol{\xi})$, mentioned in the previous section. Due to the projection property, it can be shown that $\boldsymbol{\xi} \cdot \mathbb{Q}(\boldsymbol{\xi}) = 0$. Plugging~\eqref{eq:KinEnFunc} into~\eqref{WeakForm}, then gives,
\begin{align*}
\int \mathcal{K}_s (\boldsymbol{z}) \frac{\partial f_s}{\partial t} \mathrm{d} \boldsymbol{z} &= \sum_{s, \bar{s}} \frac{1}{2} \iint\left(\frac{1}{m_{s}} \frac{\partial \mathcal{K}_s (\boldsymbol{z})}{\partial \boldsymbol{v}}-\frac{1}{m_{\bar{s}}} \frac{\partial \mathcal{K}_{\bar{s}}(\overline{\boldsymbol{z}})}{\partial \overline{\boldsymbol{v}}}\right) \cdot \mathbb{W}_{s \bar{s}}(\boldsymbol{z}, \overline{\boldsymbol{z}}) \cdot \\
&\qquad\qquad\qquad\qquad\qquad\qquad\qquad\qquad\boldsymbol{\Gamma}_{s \bar{s}}(\mathcal{S}, \boldsymbol{z}, \overline{\boldsymbol{z}}) d \overline{\boldsymbol{z}} \mathrm{d} \boldsymbol{z}, \\
&= \sum_{s, \bar{s}} \frac{1}{2} \iint\left(\boldsymbol{v} - \boldsymbol{\bar{v}} \right) \cdot \left[ \nu_{s \bar{s}} \delta(\boldsymbol{x}-\overline{\boldsymbol{x}}) f_{s}(\boldsymbol{z}) f_{\bar{s}}(\overline{\boldsymbol{z}}) \mathbb{Q}(\boldsymbol{v}-\overline{\boldsymbol{v}}) \right] \cdot \\
&\qquad\qquad\qquad\qquad\qquad\qquad\qquad\qquad\boldsymbol{\Gamma}_{s \bar{s}}(\mathcal{S}, \boldsymbol{z}, \overline{\boldsymbol{z}}) d \overline{\boldsymbol{z}} \mathrm{d} \boldsymbol{z}, \\
&= 0.
\end{align*}

Expressing $G_s (\boldsymbol{z})$ as a functional derivative,
\begin{equation}
\mathcal{G}=\sum_{s} \int G_{s} f_{s} \mathrm{d} \boldsymbol{z} \quad \Longrightarrow \quad \frac{\delta \mathcal{G}}{\delta f_{s}}=G_{s}.
\end{equation}
we can now identify $\boldsymbol{\Gamma}_{s \bar{s}} (\mathcal{G}, \boldsymbol{z}, \boldsymbol{\bar{z}})$ in the weak form~\eqref{WeakForm}, which becomes,
\begin{equation}
\int G_s (\boldsymbol{z}) \frac{\partial f_s}{\partial t} \mathrm{d} \boldsymbol{z} = \sum_{s, \bar{s}} \frac{1}{2} \iint \boldsymbol{\Gamma}_{s \bar{s}} (\mathcal{G}, \boldsymbol{z}, \boldsymbol{\bar{z}}) \cdot \mathbb{W}_{s \bar{s}}(\boldsymbol{z}, \overline{\boldsymbol{z}}) \cdot \boldsymbol{\Gamma}_{s \bar{s}}(\mathcal{S}, \boldsymbol{z}, \overline{\boldsymbol{z}}) d \overline{\boldsymbol{z}} \mathrm{d} \boldsymbol{z}.
\end{equation}
Generalizing this statement for any arbitrary functionals $\mathcal{F}$ and $\mathcal{G}$ provides a symmetric, positive semidefinite metric bracket,
\begin{equation}
(\mathcal{F}, \mathcal{G})= \sum_{s, \bar{s}} \frac{1}{2} \iint \boldsymbol{\Gamma}_{s \bar{s}} (\mathcal{F}, \boldsymbol{z}, \boldsymbol{\bar{z}}) \cdot \mathbb{W}_{s \bar{s}}(\boldsymbol{z}, \overline{\boldsymbol{z}}) \cdot \boldsymbol{\Gamma}_{s \bar{s}}(\mathcal{G}, \boldsymbol{z}, \overline{\boldsymbol{z}}) d \overline{\boldsymbol{z}} \mathrm{d} \boldsymbol{z}.
\label{MetricBracket}
\end{equation}
The collisional evolution of the arbitrary functional $\mathcal{F}$ can finally be generalized using the metric bracket into the functional differential equation,
\begin{equation}
\left.\frac{d \mathcal{F}}{d t}\right|_{\text {coll }}=\left(\mathcal{F}, \mathcal{S}\right).
\label{ColEvo}
\end{equation}

It is again straightforward to verify that,
\begin{align*}
\left( \mathcal{M}_{s} , \mathcal{G} \right) = 0, \quad \left( \mathcal{P}_{s} , \mathcal{G} \right) = 0, \quad \left( \mathcal{K}_{s} , \mathcal{G} \right) = 0,
\end{align*}
confirming that mass, momentum and kinetic energy are invariants of the metric bracket with respect to any arbitrary functional $\mathcal{G}$ and consequently the collisional evolution of these moment functionals vanishes. It is also simple to show that this metric bracket formulation satisfies basic thermodynamic principles, such as monotonicity of entropy,
\begin{equation}
\frac{\partial \mathcal{S}}{\partial t} = \left( \mathcal{S}, \mathcal{S} \right) \geq 0.
\end{equation}
This follows from the positive semidefinite nature of $\mathbb{W} (\boldsymbol{z}, \boldsymbol{\bar{z}})$. 

In the following sections, we examine discretization methods for the metric bracket~\eqref{MetricBracket} and how discrete gradients can be used to retain the desired conservation properties. For the rest of this derivation, we will consider only the single-species case which, consequently, means the species indices $s$ and $\bar{s}$ and the associated summations will be dropped.

\subsection{Marker-Particle Discretization of the Collisional Bracket}\label{sec:MetricBracketDiscretization}

Consider a common discretization for particle-in-cell simulations, primarily based on idea that the distributional density $f_s$ can be represented by the same function~\eqref{MarkPartDist} as in the Poisson bracket discretization.
To move forward with the discretization, we must first consider the local nature of the Landau operator. In practice, the two marker particles will never be at exactly the same location. Therefore, an approximation for the delta function in~\eqref{Wmat}, which is responsible for the localization of the Landau operator, must be used. For this implementation, as suggested by Hirvijoki in~\cite{Eero2020}, we will divide the spatial configuration domain into disjoint, so-called collision cells that share a boundary if adjacent. In this case, the delta function is replaced by an indicator function $\boldsymbol{1} (p, \overline{p})$ which is one if the particles $p$ and $\overline{p}$ are within the collision cell and zero otherwise.

Substituting the discrete distribution function~\eqref{MarkPartDist}, the transformation of the functional derivatives~\eqref{Identity2}, and the replacement of the localizing delta function with an indicator function into the single-species version of the metric bracket~\eqref{MetricBracket}, we get a finite-dimensional formulation of the metric bracket, 
\begin{equation}
(F, G)_{h}=\frac{1}{2} \sum_{p, \overline{p}} \boldsymbol{\Gamma}(F, p, \overline{p}) \cdot \mathbb{W}(p, \overline{p}) \cdot \boldsymbol{\Gamma}(G, p, \overline{p}),
\label{FinDimBrack}
\end{equation} 
where the vector $\boldsymbol{\Gamma}(F, p, \overline{p})$ is given by,
\begin{equation}
\boldsymbol{\Gamma}(F, p, \overline{p}) = \frac{1}{m w_{p}} \frac{\partial F}{\partial \boldsymbol{v}_p} - \frac{1}{m w_{\overline{p}}} \frac{\partial F}{\partial \boldsymbol{v}_{\overline{p}}},
\label{Gamma}
\end{equation}
and the matrix $\mathbb{W}(p, \overline{p})$ is,
\begin{equation}
\mathbb{W}(p, \overline{p}) = \nu w_{p} w_{\bar{p}} \boldsymbol{1}(p, \bar{p}) \mathbb{Q}\left(\boldsymbol{v}_{p}-\boldsymbol{v}_{\bar{p}}\right).
\end{equation}
The finite-dimensional bracket~\eqref{FinDimBrack} then leads to a new collisional evolution expression for functions that depend on particle degrees of freedom,
\begin{equation}
\left. \frac{d F}{d t} \right|_{coll} = \left( F, S_{\epsilon} \right)
\end{equation}
where $S_{\epsilon} (\boldsymbol{X}, \boldsymbol{X}; W) = \mathcal{S} \left[ f_h \right]$ is the regularized entropy functional, defined in~\eqref{RegEntropy}.

To obtain the equation of motion for the particle coordinates, we consider the specific case where $F = \boldsymbol{v}_p$. The general form of this equation of motion is,
\begin{equation}
\left.\frac{d \boldsymbol{v}_{p}}{d t}\right|_{\text {coll }}= \frac{1}{2} \sum_{\bar{p}} \boldsymbol{\Gamma}\left(\boldsymbol{v}_p, p, \bar{p}\right) \cdot \mathbb{W}(p, \overline{p}) \cdot \boldsymbol{\Gamma}\left(S_{\epsilon}, p, \bar{p}\right),
\end{equation}
Since $\boldsymbol{v}_p$ does not depend on $\boldsymbol{v}_{\overline{p}}$,
\begin{align}
\boldsymbol{\Gamma}\left(\boldsymbol{v}_p, p, \bar{p}\right) &= \frac{1 - \delta_{p \overline{p}}}{m w_{p}},
\end{align}
and this expression can be simplified to,
\begin{equation}
\left.\frac{d \boldsymbol{v}_{p}}{d t}\right|_{\text {coll }} = \frac{\nu}{m} \sum_{\bar{p}}  w_{\overline{p}} \boldsymbol{1} (p, \overline{p}) \mathbb{Q} (\boldsymbol{v}_p - \boldsymbol{v}_{\overline{p}}) \cdot \boldsymbol{\Gamma}\left(S_{\epsilon}, p, \bar{p}\right).
\end{equation}

To confirm the accuracy of this discretization, we can recheck the mass, momentum and kinetic energy conservation laws as well as the entropy monotonicity. 
It is straightforward then to show that the moments~\eqref{eq:DiscMass},~\eqref{eq:DiscMomentum}, and~\eqref{eq:DiscKinEn} are invariants of the discrete bracket. For mass and momentum this follows from the fact that $\boldsymbol{\Gamma}\left(M, p, \bar{p}\right) = 0$ and $\boldsymbol{\Gamma}\left(P, p, \bar{p}\right) = \mathbb{I} - \mathbb{I} = 0$. For kinetic energy,  $\boldsymbol{\Gamma}\left(K, p, \bar{p}\right) = \boldsymbol{v}_p - \boldsymbol{v}_{\overline{p}}$ and, consequently, $\boldsymbol{\Gamma}\left(K, p, \bar{p}\right) \cdot \mathbb{W} (p, \overline{p})= 0$. Thus, the collisional evolution equations of the moments are,
\begin{equation*}
\left.\frac{d M}{d t}\right|_{\text {coll }} = (M, S_{\epsilon})_h= 0, \quad \left.\frac{d P}{d t}\right|_{\text {coll }} = (P, S_{\epsilon})_h= 0, \quad \left.\frac{d K}{d t}\right|_{\text {coll }} = (K, S_{\epsilon})_h= 0,
\end{equation*}
and the moments are conserved. Lastly, we observe that the matrix $\mathbb{W} (p, \overline{p})$ in the discrete bracket is positive semidefinite, therefore,
the regularized entropy is trivially dissipated,
\begin{equation}
\left.\frac{d S_{\epsilon}}{d t}\right|_{\text {coll }}=\left(S_{\epsilon}, S_{\epsilon}\right)_{h} \geq 0.
\end{equation}

Now, in the proceeding section, we may consider the temporal discretization of the two brackets.


\section{Temporal Discretization}\label{sec:TemporalDiscretization}

Before addressing the time-discrete forms of the Poisson and metric brackets, it is pertinent to review the concept of discrete gradients, formally described by Gonzalez in~\cite{Gonzalez1996}. 
Consider the symplectic space $\left( V ,  \Omega \right)$ with $V$ open in $m$-dimensional Euclidian space $\mathbb{R}^m$ with points denoted by $\boldsymbol{z} = \left( \boldsymbol{z}^1, \ldots, \boldsymbol{z}^m \right)$ and $\Omega $ is a viewed as a bilinear form in the tangent space $T_{\boldsymbol{z}} V$, such that $\Omega : T_{\boldsymbol{z}} V \times T_{\boldsymbol{z}} V \rightarrow \mathbb{R}^m$, which is skew-symmetric, i.e. $\Omega (\boldsymbol{z}) = - \Omega(\boldsymbol{z})$. If we choose canonical coordinates for $\boldsymbol{z}$, the symplectic bilinear form, $\Omega(\boldsymbol{z})$ becomes a skew-symmetric matrix, $J(\boldsymbol{z})$. We replace the bilinear form with the new skew-symmetric matrix in the symplectic space, now $\left(V,J\right)$, and attach the Hamiltonian, $H:V\rightarrow \mathbb{R}^m$, to form a Hamiltonian system $(V,J,H)$ and give the continuous form of the Hamiltonian differential equations,
\begin{align}
 \dot{\boldsymbol{z}} &= X_H \left( \boldsymbol{z} \right), \label{HDE} \\
 &= J(\boldsymbol{z}) \nabla H(\boldsymbol{z}), \nonumber
\end{align}
where $X_H : V \rightarrow \mathbb{R} $ is the Hamiltonian vector field associated with the Hamiltonian function $H$.
Furthermore, we introduce the notion of a first integral, which is a smooth function $A: V \rightarrow \mathbb{R}$ that is constant along any solution and satisfies the orthogonality condition,
\begin{equation}
\nabla A(\boldsymbol{z}) \cdot X_H(\boldsymbol{z}) = 0.
\end{equation}
Note that, given the skew-symmetry of $J$, the Hamiltonian $H$ is a first integral for the Hamiltonian system $(V,J,H)$

We desire a discrete framework for numerically approximating schemes of the form,
\begin{equation}
  \frac{\boldsymbol{z}_{n+1} - \boldsymbol{z}_{n}}{\Delta t} = \boldsymbol{X}_H \left( \boldsymbol{z}_{n}, \boldsymbol{z}_{n+1} \right),
  \label{dHDE}
\end{equation}
that inherit the integral $A$, where $\boldsymbol{X}_H: V \times V \rightarrow \mathbb{R}^m$ is a two-point approximation to the exact vector field, $X_H$, e.g. $\boldsymbol{X}_H \left( \boldsymbol{z}_{n}, \boldsymbol{z}_{n+1} \right) \approx X_H \left( \boldsymbol{z}_{n+\frac{1}{2}}\right)$ and $\boldsymbol{z}_{n+\frac{1}{2}} = \frac{1}{2}\left(\boldsymbol{z}_{n}+\boldsymbol{z}_{n+1}\right)$. Assume there exists a vector $\overline{\nabla} A \left(\boldsymbol{z}_{n}, \boldsymbol{z}_{n+1}\right) \in \mathbb{R}^m$, for any $\boldsymbol{z}_{n}, \boldsymbol{z}_{n+1} \in \mathbb{R}$, that has the properties,
\begin{gather}
  \overline{\nabla} A \left( \boldsymbol{z}_{n}, \boldsymbol{z}_{n+1} \right) \approx \nabla A \left( \frac{\boldsymbol{z}_{n} + \boldsymbol{z}_{n+1}}{2} \right),\label{Prop1:a}\\
  \overline{\nabla} A \left( \boldsymbol{z}_{n}, \boldsymbol{z}_{n+1} \right) \cdot \left( \boldsymbol{z}_{n}-\boldsymbol{z}_{n+1}\right) = A\left(\boldsymbol{z}_{n}\right) - A\left(\boldsymbol{z}_{n+1}\right).\label{Prop1:b}
\end{gather}
We can therefore extend these properties to the approximation scheme such that,
\begin{align}
  A\left(\boldsymbol{z}_{n+1}\right) - A\left(\boldsymbol{z}_{n}\right) &= \overline{\nabla}A\left(\boldsymbol{z}_{n},\boldsymbol{z}_{n+1}\right) \cdot \left( \boldsymbol{z}_{n+1} - \boldsymbol{z}_{n}\right), \\
  &= \Delta t \overline{\nabla} A \left(\boldsymbol{z}_n, \boldsymbol{z}_{n+1}\right) \cdot \boldsymbol{X}_H \left(\boldsymbol{z}_n , \boldsymbol{z}_{n+1}\right),
  \label{Prop2}
\end{align}
using \eqref{dHDE}. Thus, if the approximate vector field $\boldsymbol{X}_H$ satisfies the discrete orthogonality condition,
\begin{equation}
  \overline{\nabla} A (\boldsymbol{z}_{n+1},\boldsymbol{z}_{n}) \cdot \boldsymbol{X}_H (\boldsymbol{z}_{n+1},\boldsymbol{z}_{n}) = 0, \;\;\;\; \forall \boldsymbol{z}_{n+1},\boldsymbol{z}_{n} \in V,
\end{equation}
then $A$ is an integral of~\eqref{dHDE}. We can now make the following definition:
\begin{definition}
  A \textit{discrete gradient} for a smooth function $A:V \rightarrow \mathbb{R}$ is a mapping $\overline{\nabla}: V \times V \rightarrow \mathbb{R}^m$ with the following properties.
\begin{enumerate}
  \item Directionality: $\overline{\nabla} A(\boldsymbol{z}_{n+1},\boldsymbol{z}_{n}) \cdot (\boldsymbol{z}_{n}-\boldsymbol{z}_{n+1}) = A(\boldsymbol{z}_{n}) - A(\boldsymbol{z}_{n+1})$ for any $\boldsymbol{z}_{n+1},\boldsymbol{z}_{n} \in V$.
  \item Consistency: $\overline{\nabla} A(\boldsymbol{z}_{n+1},\boldsymbol{z}_{n})=\nabla A(\frac{\boldsymbol{z}_{n+1}+\boldsymbol{z}_{n}}{2}) + \mathcal{O}(||\boldsymbol{z}_{n}-\boldsymbol{z}_{n+1}||)$ for all $\boldsymbol{z}_{n+1},\boldsymbol{z}_{n} \in V$ with $||\boldsymbol{z}_{n}-\boldsymbol{z}_{n+1}||$ sufficiently small.
  \end{enumerate}
  \label{DGDef}
\end{definition}
We attach the discrete gradient to the Hamiltonian system and refer to this new system $(V,\bar{J},\overline{\nabla},H)$ as a \textit{discrete Hamiltonian system} and associate it with the difference equation~\eqref{dHDE}, where the discrete version of the Hamiltonian vector field, $\boldsymbol{X}_H$, is now defined as,
\begin{equation}
\boldsymbol{X}_H = \bar{J}(\boldsymbol{z}_{n+1},\boldsymbol{z}_{n}) \, \overline{\nabla} H(\boldsymbol{z}_{n+1},\boldsymbol{z}_{n}),
\end{equation}
so that $H$ is a first integral due to the skew-symmetry of $\bar{J}$, now a discrete form of $J$ which approaches the continuous form in the limit $\boldsymbol{z}_{n+1}\rightarrow \boldsymbol{z}_{n}$. Now that we have a framework for the construction of a discrete gradient, we can present one potential form for the discrete gradient, often referred to as the average discrete gradient~\cite{Harten1997}.

\begin{definition}
  Let $A : \mathbb{R}^m \rightarrow \mathbb{R}$ be a smooth function. Then a discrete derivative for $A$ is defined by,
  \begin{equation}
  \overline{\nabla} A (x,y) = \int_0^1 \nabla A \left(\left(1-\xi\right)x + \xi y \right) d\xi
  \label{DGEq}
  \end{equation}
  for any two points $x, y \in \mathbb{R}^m$.
\label{DGPr}
\end{definition}

Detailed proofs of the directionality and consistency conditions for discrete gradients can be found in~\cite{Gonzalez1996,Mansfield2009}. We note that the choice of discrete gradient for the Vlasov-Poisson-Landau system is important to the success of the numerical solver, as alternative discrete gradients~\cite{Gonzalez1996,Itoh1988} contain singularities in the limit where $\left|\left|\boldsymbol{z}_n-\boldsymbol{z}_{n+1}\right|\right|\rightarrow 0$ \cite{Liu2023}. 


This approach to building first integral preserving time-integration schemes has been studied in depth for Hamiltonian systems, using various discrete gradients~\cite{Quispel1996, McLachlan1999, Norton2014, Mclaren2004, Mansfield2009}. These efforts have all been based on the simple numerical scheme, 
\begin{align}\label{SgradF}
\frac{\boldsymbol{z}_{n+1} - \boldsymbol{z}_{n}}{\Delta t} &= \boldsymbol{X}_H, \\
&= \bar{J} (\boldsymbol{z}_{n+1}, \boldsymbol{z}_{n}) \overline{\nabla} H (\boldsymbol{z}_{n+1}, \boldsymbol{z}_{n}),\nonumber
\end{align}
where $\bar{J} (\boldsymbol{z}_{n+1}, \boldsymbol{z}_{n})$ is a skew-symmetric matrix and $\overline{\nabla} H (\boldsymbol{z}_{n+1}, \boldsymbol{z}_{n})$ is the discrete gradient of the Hamiltonian, $H$, as described above. 
This formalism can be extended to the Vlasov-Poisson-Landau system through the metriplectic framework. For the collisionless dynamics, it can be shown the first integral preserving properties extend further, preserving the conservation of the moments as well as entropy. Extending the discrete gradient methodology to the Landau collision operator, however, requires some work in reformulating~\eqref{SgradF} for dissipative systems. This alternate formulation must be structure preserving as well.

\subsection{Discrete Gradient Discretization of the Poisson Bracket}

From~\eqref{eq:discPoissonBracket}, we can describe the time-discrete Poisson bracket using the discrete gradient framework as,
\begin{equation}
    \left\{ F,G \right\} \left(\boldsymbol{u}^n,\boldsymbol{u}^{n+1}\right) = \overline{\nabla} F \left(\boldsymbol{u}^n,\boldsymbol{u}^{n+1}\right)^T \bar{J}\left(\boldsymbol{u}^n,\boldsymbol{u}^{n+1}\right) \overline{\nabla} G \left(\boldsymbol{u}^n,\boldsymbol{u}^{n+1}\right),\label{eq:TimeDiscPoisson}
\end{equation}
where $\overline{\nabla} F \left(\boldsymbol{u}^n,\boldsymbol{u}^{n+1}\right)$ and $\overline{\nabla} G \left(\boldsymbol{u}^n,\boldsymbol{u}^{n+1}\right)$ are both discrete gradients of the form~\eqref{DGEq}, and $\bar{J}\left(\boldsymbol{u}^{n},\boldsymbol{u}^{n+1}\right)$ is the skew-symmetric approximation to $\mathbb{J}(\boldsymbol{u})$, introduced in Section~\ref{sec:TemporalDiscretization}. 
Like the spatially discrete form, this time-discrete form of the Poisson bracket also preserves the antisymmetry and consistency of the continuous bracket as desired. From~\eqref{eq:TimeDiscPoisson}, equations of motion in the form of~\eqref{SgradF} can be derived by replacing $F$ and $G$ with the degrees of freedom variable $\boldsymbol{u}$ and the Hamiltonian function, $H$, giving,
\begin{align}\label{eq:TimeDiscPoisson2}
  \dot{\boldsymbol{u}} &= \left\{\boldsymbol{u},H\right\} \nonumber\\
  &= \bar{J} \left(\boldsymbol{u}^n,\boldsymbol{u}^{n+1}\right) \overline{\nabla} H \left(\boldsymbol{u}^n,\boldsymbol{u}^{n+1}\right).
\end{align}
or in standard timestepping form,
\begin{align}\label{eq:DGPoisson}
  \frac{\boldsymbol{u}^{n+1} - \boldsymbol{u}^{n}}{\Delta t} = \bar{J} \left(\boldsymbol{u}^{n},\boldsymbol{u}^{n+1}\right) \overline{\nabla} H \left(\boldsymbol{u}^{n},\boldsymbol{u}^{n+1}\right).
\end{align}
This time-discrete system maintains the Casimir invariants, mass and entropy, of the spatially discrete bracket and the conservative momentum, entropy and regularized entropy. As before, the mass and entropy functions are independent of degrees of freedom, thus $\overline{\nabla} M = \overline{\nabla} S = 0$ and the two invariants remain preserved. 

To see how the momentum is conserved, consider,
\begin{align*}
\left\{ P,H \right\} \left(\boldsymbol{u}^n,\boldsymbol{u}^{n+1}\right) &= \overline{\nabla} P\left(\boldsymbol{u}^n,\boldsymbol{u}^{n+1}\right)^T \bar{J}\left(\boldsymbol{u}^n,\boldsymbol{u}^{n+1}\right) \overline{\nabla} H \left(\boldsymbol{u}^n,\boldsymbol{u}^{n+1}\right),
\end{align*}
where the discrete gradient of $P$ is the constant vector,
\begin{equation*}
    \overline{\nabla} P = \nabla P = \begin{pmatrix} \nabla_{\boldsymbol{X}} P \\ \nabla_{\boldsymbol{V}} P    \end{pmatrix} = \begin{pmatrix} 0 \\ \mathbb{M}p \end{pmatrix}.
\end{equation*}
Then, using the skew-symmetric time-discrete bracket $\bar{J}$ (where we assume the canonical structure $\bar{J}{XV} = \mathbb{M}p^{-1}$), the time discrete bracket becomes:
\begin{align*}
    \left\{ P,H \right\} &= \overline{\nabla} P^T \bar{J} , \overline{\nabla} H \\
    &= \begin{pmatrix} \mathbf{0} & \mathbb{M}p \end{pmatrix} \begin{pmatrix} \mathbf{0} & \mathbb{M}p^{-1} \\ - \mathbb{M}p^{-1} & \mathbf{0} \end{pmatrix} \begin{pmatrix} \overline{\nabla}_{\boldsymbol{X}} H \\ \overline{\nabla}_{\boldsymbol{V}} H \end{pmatrix} \\
    &= \begin{pmatrix} - \mathbf{1}^T & \mathbf{0} \end{pmatrix} \begin{pmatrix} \overline{\nabla}_{\boldsymbol{X}} H \ \overline{\nabla}_{\boldsymbol{V}} H \end{pmatrix} \\
    &= - \mathbf{1}^T \overline{\nabla}_{\boldsymbol{X}} H
\end{align*}
where the second-to-last step is derived from the matrix multiplication $\mathbb{M}_p (- \mathbb{M}_p^{-1}) = - \mathbb{I}$, and $\mathbf{1}^T$ is a row vector of ones. The term $-\mathbf{1}^T \overline{\nabla}_{\boldsymbol{X}} H$ is equivalent to the sum of all position-component gradients of the Hamiltonian. As shown for the semi-discrete bracket, momentum conservation relies on the condition that the total force is zero, i.e. $\sum_p \nabla_{\boldsymbol{x}_p} H = \mathbf{0}$. By the consistency property of the discrete gradient (see Definition~\ref{DGDef}), the sum of discrete gradients of the Hamiltonian reduces to ,
\begin{align*}
\mathbf{1}^T \overline{\nabla}_{\boldsymbol{X}} H &= \sum_p \overline{\nabla}_{\boldsymbol{x}_p} H,\\
&= \sum_p \nabla_{\boldsymbol{x}_p} H + \mathcal{O} \left(||\boldsymbol{u}^{n+1}-\boldsymbol{u}^{n}||\right),\\
&= \mathcal{O} \left(||\boldsymbol{u}^{n+1}-\boldsymbol{u}^{n}||\right).
\end{align*}
Therefore, $\left\{ P,H \right\} \sim \mathbf{0}$ and the momentum is still conserved up to the order of the error in the change in the state vector between steps.

We may similarly show this for regularized entropy, however, it is simpler to show that regularized entropy is conserved using just the consistency property of the discrete gradient. The time-discrete update equation for momentum is given by,
\begin{equation*}
    \frac{S_{\varepsilon}^{n+1} - S_{\varepsilon}^{n}}{\Delta t} = \overline{\nabla} S_{\varepsilon}^T \, \bar{J} \, \overline{\nabla} H .
\end{equation*}
By the consistency property of the discrete gradient this equation reduces to,
\begin{align*}
    \frac{S_{\varepsilon}^{n+1} - S_{\varepsilon}^{n}}{\Delta t} &= \nabla S_{\varepsilon}^T \,\mathbb{J} \, \nabla H  + \mathcal{O} \left( ||\boldsymbol{u}^{n+1} - \boldsymbol{u}^n|| \right),\\
    &= \mathcal{O} \left( ||\boldsymbol{u}^{n+1} - \boldsymbol{u}^n|| \right),
\end{align*}\
as we have already shown $\left\{S_{\varepsilon},H\right\}=0$ for the finite-dimensional bracket. Thus, regularized entropy  is conserved up to order $\mathcal{O} \left( ||\mathbf{u}^{n+1} - \mathbf{u}^n|| \right)$. Lastly, by building the time-discrete framework around the Hamiltonian, $H$, conservation of $H$ is guaranteed through the discrete gradient definition~\cite{Kraus2017,Li2023}.

It is worth noting that this time-discrete system is rather unique in its mass, momentum, total energy, and entropy conserving qualities. In previous kinetic frameworks~\cite{BirdsallLangdon,Esirkepov2001,Markidis2011}, one or more of these qualities is usually sacrificed in the discretization, typically either momentum or energy as mass and entropy conservation are trivial. In this discretization, the correct choice in a finite element basis provides us with the desired momentum consistency, whereas, the use of discrete gradients preserves the antisymmetry of the Poisson bracket, and thus, the conservation of total energy. Furthermore, we have introduced a new form of entropy, regularized entropy, which will be important in future sections, and shown that it is similarly conserved in the chosen framework.  


\subsection{Discrete Gradient Dependent Integrators for the Metric Bracket}\label{sec:DGDI}

A simple temporal discretization of \eqref{FinDimBrack} is introduced by Carrillo et al.~\cite{Carrillo2020}. This time integrator, however, fails to conserve energy or preserve the monotonicity of the entropy. To guarantee this preservation of entropy monotonicity, we follow from the work presented by Hirvijoki in~\cite{Eero2020} and introduce a new timestepping integrator that uses discrete gradients, defined by Gonzalez in~\cite{Gonzalez1996},  
\begin{equation}
\frac{\boldsymbol{v}_{p}^{n+1}-\boldsymbol{v}_{p}^{n}}{\Delta t}=\frac{\nu}{m} \sum_{\bar{p}} w_{\bar{p}} \boldsymbol{1}(p, \bar{p}) \mathbb{Q}\left(\overline{\boldsymbol{\Gamma}_{n}^{n+1}(K, p, \bar{p})}\right) \cdot \overline{\boldsymbol{\Gamma}_{n}^{n+1}\left(S_{\epsilon}, p, \bar{p}\right)},
\label{DGColV}
\end{equation}
where the vector $\overline{\boldsymbol{\Gamma}^{n+1}_{n} (F,p,\overline{p})}$ is,
\begin{equation}
\overline{\boldsymbol{\Gamma}_{n}^{n+1}(F, p, \overline{p})}=\frac{1}{m w_{p}} \overline{\nabla} F (\boldsymbol{v}_{p}^{n},\boldsymbol{v}_{p}^{n+1})-\frac{1}{m w_{\overline{p}}} \overline{\nabla} F (\boldsymbol{v}_{\overline{p}}^{n},\boldsymbol{v}_{\overline{p}}^{n+1}).
\label{GammaBar}
\end{equation}
As in previous sections, we utilize average discrete gradient~\eqref{DGEq}, also used in~\cite{Eero2020}, for $\overline{\nabla} F (v_{p}^{n},v_{p}^{n+1})$, rather than Gonzalez's midpoint discrete gradient. This choice avoids any singularities encountered in slow-moving plasmas (or simulations with short time steps). Further generalizing, the collisional discrete-time evolution of an arbitrary functional is given by,
\begin{equation}
\frac{F^{n+1} - F^{n}}{\Delta t} = \frac{1}{2} \sum_{p,\overline{p}} \overline{\boldsymbol{\Gamma}^{n+1}_{n} (F,p,\overline{p})} \cdot \overline{W^{n+1}_{n} (p,\overline{p})} \cdot \overline{\boldsymbol{\Gamma}^{n+1}_{n} (S_{\epsilon},p,\overline{p})},
\label{DGColDTE}
\end{equation}
where the matrix $\overline{W^{n+1}_{n} (p,\overline{p})}$ is,
\begin{equation}
\overline{W_{n}^{n+1}(p, \overline{p})}=\nu w_{p} w_{\overline{p}} \boldsymbol{1}(p, \overline{p}) \mathbb{Q}\left(\overline{\boldsymbol{\Gamma}_{n}^{n+1}(K, p, \overline{p})}\right).
\label{WDGmat}
\end{equation}
Returning to the discussion of discrete gradients, we now recognize that in~\eqref{DGColDTE}, and specifically in~\eqref{DGColV}, we have derived a new integrator for the metric bracket in similar form to the numerical scheme~\eqref{dHDE}, proposed by Gonzalez. 
This so-called \textit{discrete gradient dependent integrator} (DGDI) is not itself a discrete gradient, as it does not pass the directionality test, but relies on the properties of the discrete gradient~\eqref{DGEq} (see Definition~\eqref{DGDef}) to be monotonic. It is, therefore, worthwhile to prove that our new integrator remains consistent with the finite-dimensional form of the collisional evolution~\eqref{ColEvo} and to show that the integrator conserves the moments and preserves the monotonicity of entropy.

\begin{theorem}
The collisional discrete-time evolution equation, given by~\eqref{DGColDTE}, is a consistent approximation to the finite-dimensional bracket~\eqref{FinDimBrack}. \\
\begin{equation}
\left. \frac{d F}{d t}\right|_{\text {coll }} \approx \frac{F^{n+1} - F^{n}}{\Delta t},
\end{equation}
in the limit where $||\boldsymbol{v}_{p}^{n+1} - \boldsymbol{v}_{p}^{n}|| \rightarrow 0$ and $||\boldsymbol{v}_{\overline{p}}^{n+1} - \boldsymbol{v}_{\overline{p}}^{n}|| \rightarrow 0$.
\end{theorem}

\begin{proof}
First we consider what happens to~\eqref{GammaBar} as $||\boldsymbol{v}_{p}^{n+1} - \boldsymbol{v}_{p}^{n}||$ and $||\boldsymbol{v}_{\overline{p}}^{n+1} - \boldsymbol{v}_{\overline{p}}^{n}||$ approach zero, 
\begin{align*}
\overline{\boldsymbol{\Gamma}_{n}^{n+1}(F, p, \overline{p})}=&\frac{1}{m w_{p}} \overline{\nabla} F (\boldsymbol{v}_{p}^{n},\boldsymbol{v}_{p}^{n+1})-\frac{1}{m w_{\overline{p}}} \overline{\nabla} F (\boldsymbol{v}_{\overline{p}}^{n},\boldsymbol{v}_{\overline{p}}^{n+1})\\
= &\frac{1}{m w_{p}} \left( \nabla F \left( \boldsymbol{v}_{p}^{n+1/2} \right) + \mathcal{O} (||\boldsymbol{v}_{p}^{n+1} - \boldsymbol{v}_{p}^{n}||) \right) \\
&- \frac{1}{m w_{\overline{p}}} \left( \nabla F \left( \boldsymbol{v}_{\overline{p}}^{n+1/2} \right) + \mathcal{O} (||\boldsymbol{v}_{\overline{p}}^{n+1} - \boldsymbol{v}_{\overline{p}}^{n}||) \right)\\
\approx &\frac{1}{m w_{p}} \nabla F \left( \boldsymbol{v}_{p}^{n+1/2} \right) - \frac{1}{m w_{\overline{p}}} \nabla F \left( \boldsymbol{v}_{\overline{p}}^{n+1/2} \right)\\
\approx & \boldsymbol{\Gamma} (F, p, \overline{p})
\end{align*}
where $\boldsymbol{v}_p^{n+1/2}=\frac{1}{2}\left(\boldsymbol{v}_p^{n+1}+\boldsymbol{v}_p^{n}\right)$ and likewise for $\boldsymbol{v}_{\bar{p}}^{n+1/2}$. Furthermore, we recall the consistency property in Definition~\eqref{DGDef} to obtain the midpoint approximation, $\nabla F(\boldsymbol{v}_p^{n+1/2})$, of the discrete gradient, $\overline{\nabla} F (\boldsymbol{v}_{p}^{n},\boldsymbol{v}_{p}^{n+1})$. Plugging this into relation into~\eqref{DGColDTE} gives,
\begin{align*}
\frac{F^{n+1} - F^{n}}{\Delta t} &= \frac{1}{2} \sum_{p,\overline{p}} \overline{\boldsymbol{\Gamma}^{n+1}_{n} (F,p,\overline{p})} \cdot \overline{W^{n+1}_{n} (p,\overline{p})} \cdot \overline{\boldsymbol{\Gamma}^{n+1}_{n} (S_{\epsilon},p,\overline{p})} \\
&\approx \frac{1}{2} \sum_{p,\overline{p}} \boldsymbol{\Gamma}^{n+1}_{n} (F,p,\overline{p}) \cdot \left[ \nu w_{p} w_{\overline{p}} \boldsymbol{1}(p, \overline{p}) \mathbb{Q}\left(\boldsymbol{\Gamma}_{n}^{n+1}(K, p, \overline{p})\right) \right] \cdot \boldsymbol{\Gamma}^{n+1}_{n} (S_{\epsilon},p,\overline{p})
\end{align*}
Realizing $\boldsymbol{\Gamma}(K,p,\overline{p}) = \boldsymbol{v}_p - v_{\overline{p}}$, we can therefore show that,
\begin{align*}
\frac{F^{n+1} - F^{n}}{\Delta t} &\approx \frac{1}{2} \sum_{p,\overline{p}} \boldsymbol{\Gamma}^{n+1}_{n} (F,p,\overline{p}) \cdot \left[ \nu w_{p} w_{\overline{p}} \boldsymbol{1}(p, \overline{p}) \mathbb{Q}\left( \boldsymbol{v}_p - v_{\overline{p}}\right) \right] \cdot \boldsymbol{\Gamma}^{n+1}_{n} (S_{\epsilon},p,\overline{p}) \\
&\approx \frac{1}{2} \sum_{p,\overline{p}} \boldsymbol{\Gamma}^{n+1}_{n} (F,p,\overline{p}) \cdot \mathbb{W}(p, \overline{p}) \cdot \boldsymbol{\Gamma}^{n+1}_{n} (S_{\epsilon},p,\overline{p})\\
&\approx  \left. \frac{d F}{d t}\right|_{\text {coll }}
\end{align*}
\end{proof}

Proving conservation of the moments follows similarly from the conservation proofs in previous sections. Conservation of mass is trivial to show for the DGDI as the discrete mass function is a function of the weights alone. Thus $\boldsymbol{\Gamma}^{n+1}_{n}(M^n, p, \overline{p}) = 0$ and the conservation of~\eqref{DGColDTE} follows. Conservation of momentum again follows from the condition $\overline{\boldsymbol{\Gamma}^{n+1}_{n}(P, p, \overline{p})} = \mathbb{I} - \mathbb{I} = 0$,
\begin{align*}
\frac{P^{n+1}-P^{n}}{\Delta t} &= \frac{\nu}{2} \sum_{p, \overline{p}} \overline{\boldsymbol{\Gamma}_{n}^{n+1}(P, p, \overline{p})} \cdot \overline{W_{n}^{n+1}(p, \overline{p})} \cdot \overline{\boldsymbol{\Gamma}_{n}^{n+1}\left(S_{\epsilon}, p, \overline{p}\right)}, \\
&= \frac{\nu}{2} \sum_{p, \overline{p}} (\mathbb{I} - \mathbb{I}) \cdot \overline{W_{n}^{n+1}(p, \overline{p})} \cdot \overline{\boldsymbol{\Gamma}_{n}^{n+1}\left(S_{\epsilon}, p, \overline{p}\right)}, \\
&= 0.
\end{align*}
Conservation of kinetic energy is dependent on the fact that $\mathbb{Q}(\boldsymbol{\xi})$ has the null eigenvector $\boldsymbol{\xi}$, thus $\overline{\boldsymbol{\Gamma}_{n}^{n+1}(K^n, p, \overline{p})} \cdot \mathbb{Q}(\overline{\boldsymbol{\Gamma}_{n}^{n+1}(K^n, p, \overline{p})}) = 0$. Therefore, the discrete-time collisional rate-of-change of kinetic energy follows as such,
\begin{align*}
\frac{K^{n+1}-K^{n}}{\Delta t}&=\frac{\nu}{2} \sum_{p, \overline{p}} \overline{\boldsymbol{\Gamma}_{n}^{n+1}(K^n, p, \overline{p})} \cdot \overline{W_{n}^{n+1}(p, \overline{p})} \cdot \overline{\boldsymbol{\Gamma}_{n}^{n+1}\left(S_{\epsilon}, p, \overline{p}\right)} \\
&=\frac{\nu}{2} \sum_{p, \overline{p}} \overline{\boldsymbol{\Gamma}_{n}^{n+1}(K^n, p, \overline{p})} \cdot \left[ \nu w_{p} w_{\overline{p}} \boldsymbol{1}(p, \overline{p}) \mathbb{Q}\left(\overline{\boldsymbol{\Gamma}_{n}^{n+1}(K^n, p, \overline{p})}\right) \right] \cdot \overline{\boldsymbol{\Gamma}_{n}^{n+1}\left(S_{\epsilon}, p, \overline{p}\right)} \\
&= 0
\end{align*}
Lastly, we consider the preservation of entropy monotonicity. As a result of $\overline{W_{n}^{n+1}(p, \overline{p})}$ being a positive semidefinite matrix, the discrete-time entropy behaves as desired,
\begin{align*}
\frac{S_{\epsilon}^{n+1}-S_{\epsilon}^{n}}{\Delta t}&=\frac{\nu}{2} \sum_{p, \overline{p}} \overline{\boldsymbol{\Gamma}_{n}^{n+1}(S_{\epsilon}, p, \overline{p})} \cdot \overline{W_{n}^{n+1}(p, \overline{p})} \cdot \overline{\boldsymbol{\Gamma}_{n}^{n+1}\left(S_{\epsilon}, p, \overline{p}\right)}, \\
&\geq 0.
\end{align*}

We now have a fully time discrete metric bracket that is preserves all the moments and entropy monotonicity. In the next section, we will consider a numerical implementations of the derived discretized Poisson and metric brackets and test them on classic plasma physics tests.


\section{\label{sec:NumImp}Numerical Implementation}

The numerical implementation presented in this section is based on the \texttt{PETSc} Particle-In-Cell (\texttt{PETSc-PIC}) framework~\cite{Finn2023}. The \texttt{PETSc-PIC} framework provides scalable and structure-preserving algorithms for both the Poisson and metric brackets. The goal of \texttt{PETSc} is to provide composable pieces from which optimal simulations can be constructed. \texttt{PETSc} user level APIs allow applications to delay implementation choices, such as solver details, until runtime using dynamic configuration~\cite{BrownKnepleySmith14}. Recent advances in the \texttt{PETSc-PIC} framework have included conservative projections between the finite element and particle bases~\cite{Pusztay2022}, as well as a conservative remapping step for reducing stochastic noise growth in long time simulations~\cite{Adams2025}. For this work, the remapping phase was ignored as its effects on the microscale physics of the kinetics system are an active area of study (see Section 8 of~\cite{Adams2025}). The implementations for the collisionless Vlasov-Poisson and collisional Landau systems each possess unique conservative qualities that will be discussed in more detail in Section~\ref{sec:collisionless} and Section~\ref{sec:collisional}.


\subsection{Solver Considerations}\label{sec:SolverConsiderations}
In previous sections, the structure-preserving properties of both the Poisson and metric brackets have been discussed in detail. As an implicit timestepping method, the discrete gradient integrator requires the implementation to construct and solve a nonlinear system of equations at each time step. The complexity of the problem means that care is needed in the choice of a solver for the system. The default choice for nonlinear solves in the \texttt{PETSc} library is a standard Newton-Raphson method with line search. However, the standard Newton-Raphson method requires a user-defined Jacobian matrix, which in both the Vlasov-Poisson and Landau systems is difficult to derive. Previous works~\cite{Hu2024} have avoided this task by utilizing fixed point iteration solvers, which don't require Jacobian matrices. While effective, these fixed point methods are linear in convergence, and thus very slow to a solution, particularly with larger time step sizes. 

An alternative to fixed point iteration methods are quasi-Newton methods, which use an approximation of the Jacobian, starting with a positive-definite identity matrix and updated each iteration. Quasi-Newton methods are advantageous in large-scale, kinetic frameworks as they are superlinear, under the right conditions, and therefore require far less iterations to achieve nonlinear convergence. While each iteration is on average more expensive than a fixed point iteration, the efficiency increase is still apparent. Moreover, a limited memory Broyden-Fletcher-Goldfarb-Shanno (L-BFGS)~\cite{Broyden1970,Nocedal1980} method may be utilized  to reduce memory usage of the algorithm, further improving efficiency.

The use of these inexact nonlinear solvers, however, comes at the cost of structure-preservation. Principally, two crucial pillars of the preserved structure in the spatially- and time-discrete bracket are lost when using quasi-Newton methods. The first is the exact satisfaction of the discrete chain rule, i.e. the directionality condition in the discrete gradient definition~\ref{DGDef}. If we consider, for the collisionless case, that the nonlinear solver only \textit{approximately} enforces the nonlinear residual,
\begin{equation}
    r = \boldsymbol{u}^{n+1} - \boldsymbol{u}^{n} - \Delta t \bar{J} \, \overline{\nabla} H = 0
\end{equation}
then the directionality condition no longer holds exactly. Even for small residuals, a drift on the order $O(||r||)$ is introduced to the total energy, $H$, as,
\begin{align*}
    H(\boldsymbol{u}^{n+1}) - H(\boldsymbol{u}^{n}) &\approx \overline{\nabla} H \cdot \left(\boldsymbol{u}^{n+1} - \boldsymbol{u}^{n}\right), \\
    &=\Delta t \overline{\nabla} H \bar{J} \, \overline{\nabla} H + \overline{\nabla} H^T r,\\
    &=\overline{\nabla} H^T r.
\end{align*}
In general for any conserved quantity, $A$, $A(\boldsymbol{u}^{n+1}) - A(\boldsymbol{u}^{n}) = \overline{\nabla} A^T r$, and thus we expect some drift controlled by the tolerance levels of the solver.

The second lost pillar of structure-preservation is the exact skew-symmetry of the Jacobian. A true Newton-Raphson method would assemble the Jacobian matrix, which because of the properties of $\bar{J}$ and $\overline{\nabla} H$, would preserve the antisymmetry of the system, and consequently the conservative properties. However, as discussed the L-BFGS algorithm builds a positive-definite low-rank approximation to the Jacobian which does not retain the principle qualities of the true Jacobian required to maintain the antisymmetry of the system.

Fixed point iteration methods are also subject to similar conservation breaking issues as they do not preserve the geometry of the discretized brackets. The accuracy and conservative properties of both fixed point and quasi-Newton methods can ideally be controlled with solver tolerances, however, this remains to be studied in depth for this problem. We will briefly explore here the effect of tolerance level on structure preservation using the collisionless tests, as they are faster and more efficient than the Landau collision tests and therefore a better testing ground for this. 

\subsection{Collisionless Dynamics}\label{sec:collisionless}
We first consider an implementation of discrete gradient integrators for the collisionless Vlasov-Poisson system. We use the collisionless system not only as a test of the structure-preserving qualities of our discretization, but as a useful test of nonlinear solvers. As discussed in Section~\ref{sec:PoissonBracket}, in an ideal case, we would expect the system's mass, momentum, energy and entropy to be conserved as they are in the time-discrete form (with some $\mathcal{O}(||\Delta u||)$ error for momentum and regularized entropy). However, we must make compromises in the selection of the nonlinear solver for computational efficiency which sacrifices some of the structure-preserving aspects of our formulation.
The collisionless Vlasov-Poisson system may be formulated using discrete gradients where the skew-symmetric matrix is,
\begin{equation}
\bar{J} = \begin{pmatrix} 0 & 1 \\ -1 & 0\end{pmatrix}
\end{equation}
and the Hamiltonian function is as defined in~\eqref{eq:DiscHamiltonian}.

We test the collisionless algorithm in one dimension (1X-1V) on the standard plasma test, Landau damping. In this test, an initial perturbation in the electric field is damped out of the particle system without collisions at a calculable rate. A detailed overview of the analytic derivation of the electric field oscillation frequency and damping rate can be found in~\cite{Finn2023}. We will use the same system parameters and, thus, expect an electric field oscillation frequency and damping rate of $\omega_r=1.416$ and $\gamma=-0.153$, respectively. The plasma particles are given an initial weight distribution of,
\begin{gather}
  f\left(x,v,t=0\right) = \frac{1}{\sqrt{2\pi v_{th}^2}} e^{-v^2/2 v_{th}^2} \left(1 + \alpha \cos \left(kx\right)\right), \\
  \left(x,v\right) = \left[0,2\pi/k\right] \times \left[-v_{max},v_{max}\right], \nonumber
\end{gather}
where $v_{th} = 1$, $\alpha=0.01$, $k=0.5$, $v_{max}=10$ and the boundaries are periodic. For these tests, we choose a grid consisting of 80 spatial cells and 6,000 particles per cell. This is a slight decrease in resolution from previous work, however, the introduction of the regularized entropy monitor, which is $\mathcal{O}(N^2)$, necessitates the use of smaller resolutions for reasonable simulation run times. We use a PIC time step size of $\Delta t=0.01$, which in previous convergence work was found to be ideal for each time stepper. For this first test, the relative tolerance for the L-BFGS solver was set to its lowest level, $\mathcal{O} (10^{-11})$, ensuring maximum structure-preservation. The symplectic and Runge-Kutta integrators are explicit, and thus avoid the same pitfalls introduced by nonlinear solvers. However, these explicit methods introduce their own errors and break structure-preservation in the discretization phase. The Runge-Kutta integrator is included here as a comparison for a non-conservative method. 
Symplectic integrators preserve the symplectic structure of the system and the associated Noether invariants arising from continuous symmetries, while Casimir invariants are conserved independently due to degeneracies in the Poisson bracket. Furthermore, by preserving the symplectic two-form, these integrators ensure that the energy and momentum errors remain bounded over long-time simulations, even though these quantities are not exactly conserved at each time step.
These differences present tradeoffs we explore in the tests.
\begin{figure}[!h]
  \begin{center}
    \includegraphics[width=0.9\columnwidth]{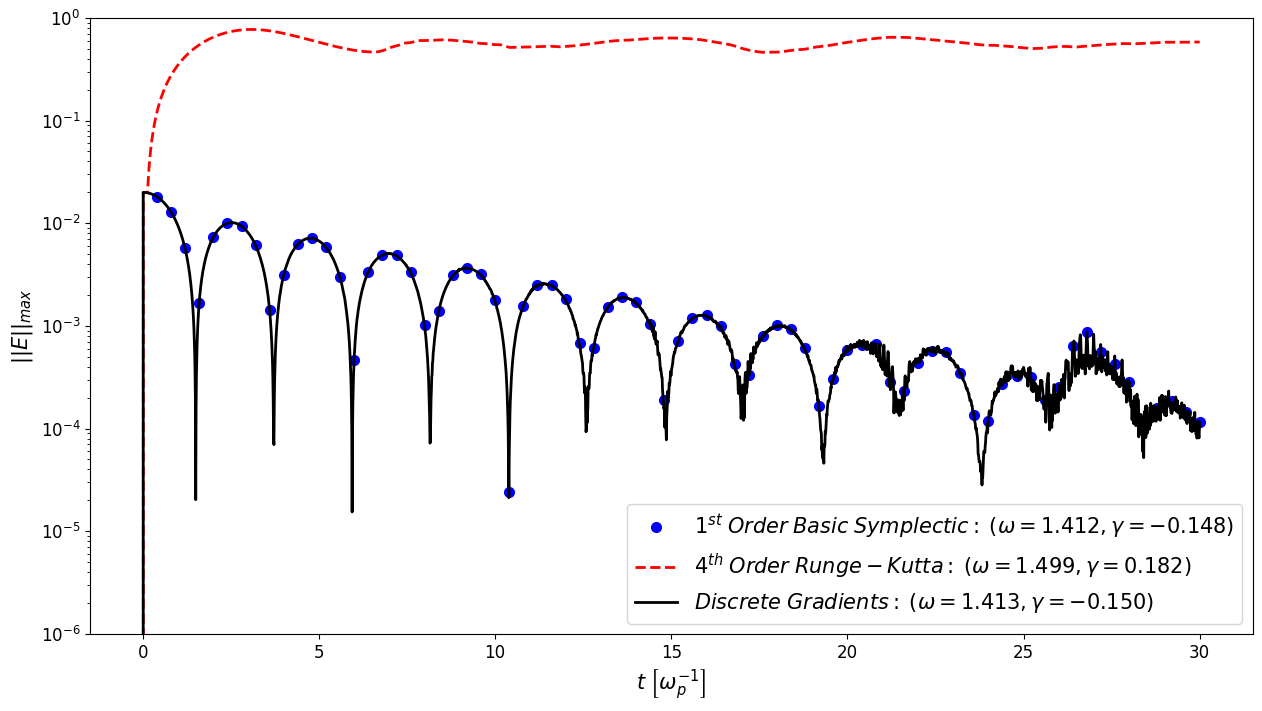}
  \end{center}
  \caption{The max-norm of the electric field as a function of time for three different timestepping integrators: discrete gradient integrators, 1st-order symplectic, and 4th-order Runge-Kutta. For this test, the relative tolerance for the discrete gradient methods was set to $10^{-11}$.}
  \label{fig:LD_TSComp}
\end{figure}

\begin{figure}[!h]
  \begin{center}
    \includegraphics[width=0.9\columnwidth]{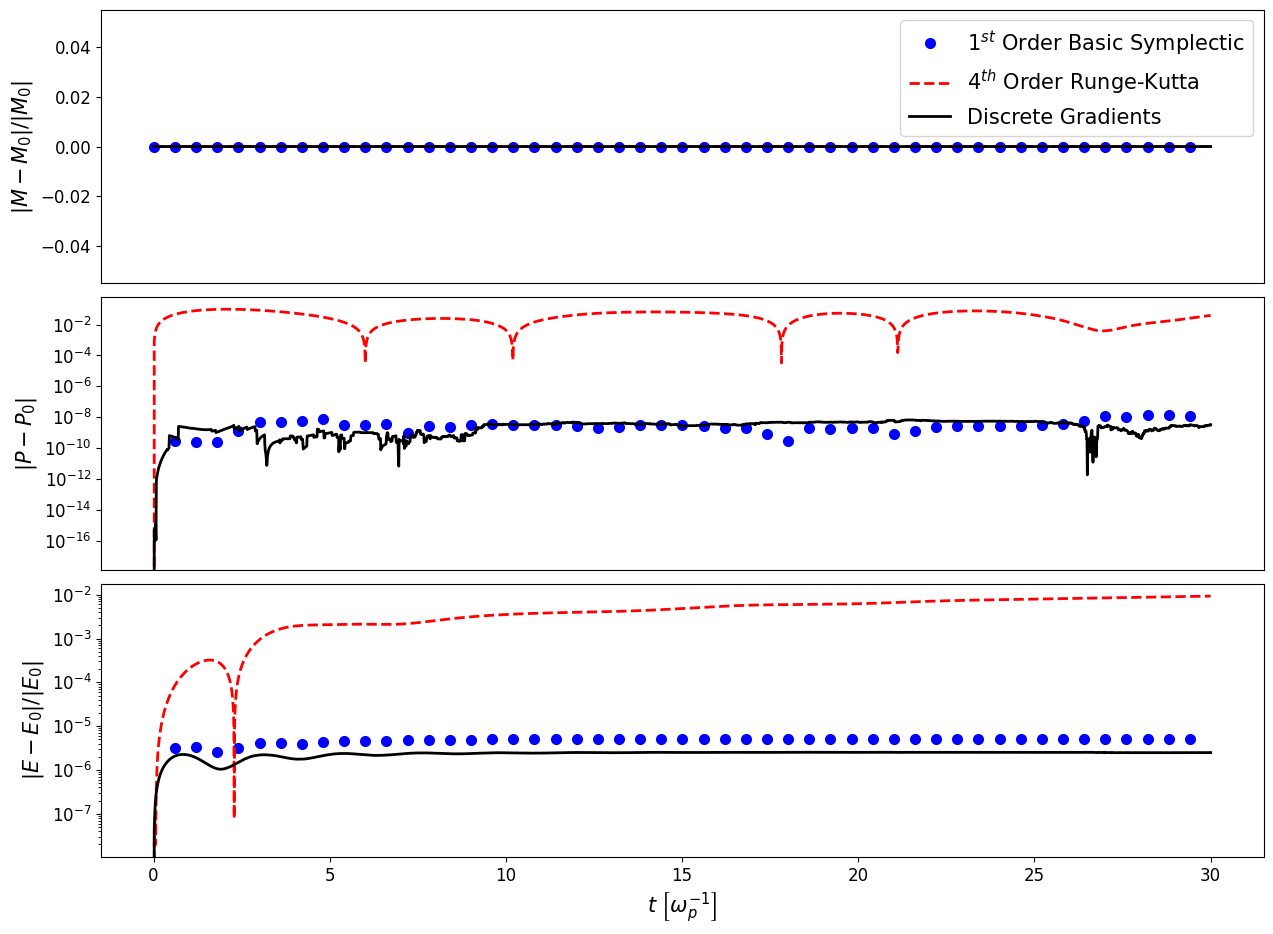}
  \end{center}
  \caption{The moments, mass, momentum and energy, as a function of time for three different timestepping integrators. Absolute error is used for the momentum plot as the true value of momentum is zero, and therefore, the relative error is non-convergent.}
  \label{fig:LD_MomComp}
\end{figure}

Integrals appearing in the average discrete gradient~\eqref{DGEq} are evaluated using standard tensor-product Gauss-Legendre quadratures. Gauss-Legendre quadratures with $m$ nodes are exact for all polynomials up to degree $2m-1$. As both kinetic and field components of the Hamiltonian function~\eqref{eq:DiscHamiltonian} are quadratic, a 2-point Gauss-Legendre rule is sufficient to ensure exact calculation of the discrete gradient and retention of it's structure preserving qualities. In contrast, the regularized entropy function~\eqref{eq:DiscRegEntropy} is evaluated using a Gauss-Hermite quadrature, which are designed for functions of the form $\int f(z) e^{-z^2} \mathrm{d} \boldsymbol{z}$. This form is not exactly like~\eqref{eq:DiscRegEntropy}, however, in tests it has proven accurate for the purpose of this application.

Figure~\ref{fig:LD_TSComp} shows the maximum value of the electric field, $E_{max} = \max_{\Omega} \|E\|$, over time. The values for $\gamma$ and $\omega_r$ were measured by fitting the max norm of the exact equation for the electric field evolution in time to the data, given by,
\begin{equation}
  E_{exact} \left(x,t\right) = 4 A r e^{\gamma_{exact} t} \sin \left(k x\right) \cos \left(\omega_{r,exact} t - \phi_0\right),
\end{equation}
where $\left(A,r,\phi_0\right)=\left(0.01,0.3677,0.536245\right)$. We ignore the peaks at $t = 0\left[\omega_p^{-1}\right]$ and past $t = 25\left[\omega_p^{-1}\right]$ as they may represent complex frequency roots other than the desired \textit{Landau root} or data corrupted by particle noise. The symplectic and discrete gradient integrators accurately capture the damping behavior of the electric field with a slight improvement in accuracy when using the discrete gradient integrator ($error_{DG} = (0.212\%,1.961\%)$ and $error_{BSI} = (0.282\%,2.614\%)$). As expected, the non-conservative Runge-Kutta integrator fails to capture the electric field damping at any level.
\begin{figure}[h!]
  \begin{center}
    \includegraphics[width=0.9\columnwidth]{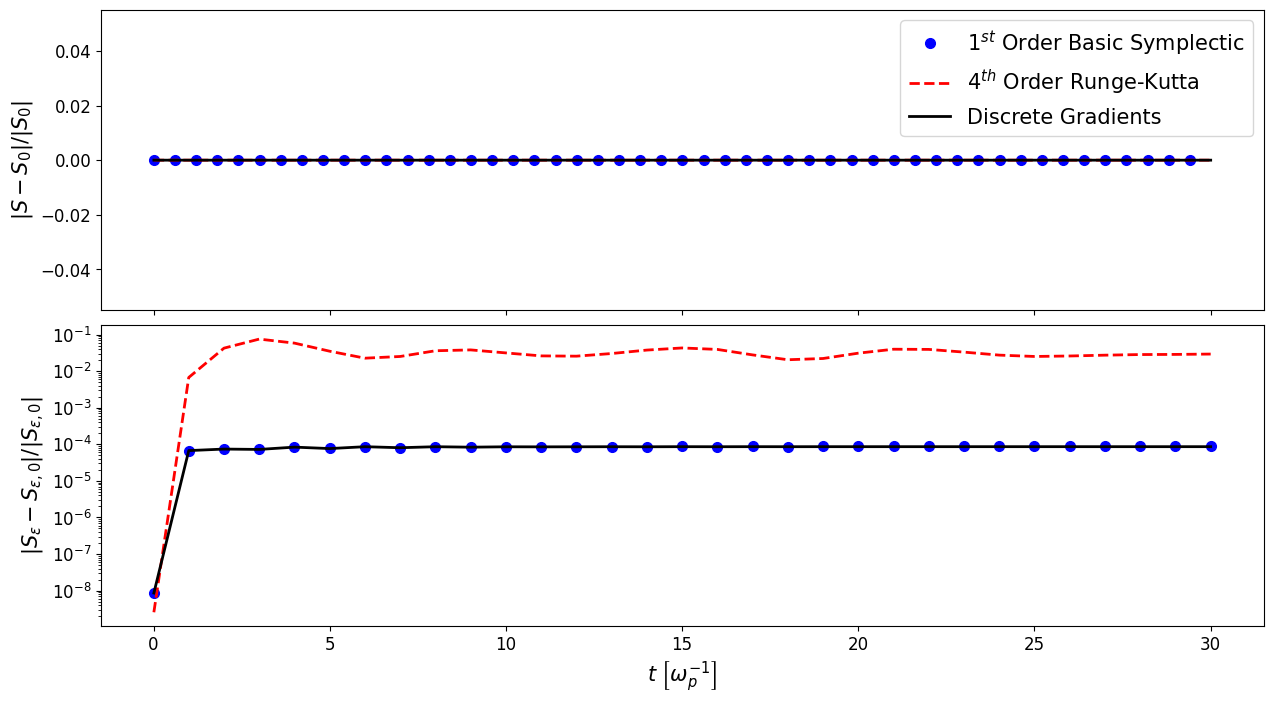}
  \end{center}
  \caption{Entropy and regularized entropy for each of the three different timestepping integrators. Regularized entropy is only computed every 100 steps as it is extremely computationally expensive ($\mathcal{O} (N^2)$).}
  \label{fig:LD_Entropy}
\end{figure}
\begin{figure}[h!]
  \begin{center}
    \includegraphics[width=0.9\columnwidth]{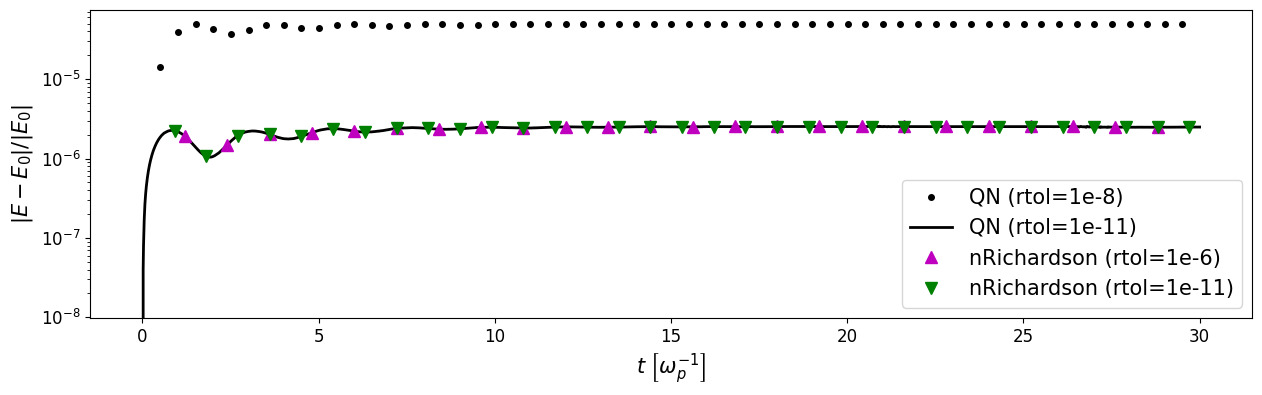}
  \end{center}
  \caption{A comparison of energy conservation for two nonlinear solver types: quasi-Newton (L-BFGS) and nonlinear Richardson (fixed-point). Each solver was tested with two relative tolerance levels. The nonlinear Richardson solver was given a larger tolerance difference as the results showed no change with the original tolerance difference.}
  \label{fig:LD_SolvComp}
\end{figure}

The error in the moments, mass, momentum and energy, over time (Figure~\ref{fig:LD_MomComp}), reveals that all three timestepping methods are exactly mass conservative, while only the symplectic and discrete gradient integrators achieve relatively low levels of momentum and total energy conservation, remaining bounded over time. Mass conservation is relatively trivial as the number of particles and total weight remains constant throughout the simulation, however, with the mapping stages between particle and finite element spaces, it is worthwhile to confirm this property. Furthermore, each timestepping method exactly conserves the entropy of the system while the regularized entropy is only kept stable at $\sim \mathcal{O} (10^{-4})$ by the discrete gradient integrators and symplectic methods (Figure~\ref{fig:LD_Entropy}). In an ideal, continuous model, we expect the discrete gradient integrator to achieve full conservation of all the moments, as well as entropy, to machine precision, and in fact we do see some improvement over the symplectic methods. However, as expected, limitations in the L-BFGS algorithm have caused a loss of conservation of energy, momentum and regularized entropy. The regularized entropy has an $\mathcal{O} (N^2)$ dependence on the particles, and thus may additionally be degraded by the introduction of any particle noise.

In Figure~\ref{fig:LD_SolvComp}, the L-BFGS method is compared to a classic fixed point iteration solver (nonlinear Richardson), each at two different relative solver tolerance levels. As shown in the figure, reducing the tolerance of the L-BFGS solver improves the conservative nature of the algorithm but does not have an effect on the fixed point iterator. The fixed point method was further tested with a larger tolerance difference than the L-BFGS solver with no change in results. It is unknown why the fixed point method showed no improvement when tolerances were lowered but this is an interesting topic for future studies. A third test with a relative tolerance level of $\mathcal{O} (10^{-14})$ was attempted for the L-BFGS solver, however, the solver stalled out at a relative norm of $~\mathcal{O} (10^{-13})$, eventually reaching the maximum iteration limit and failing. It is likely that the Jacobian for the L-BFGS method is no longer accurate enough past this point and convergence is breaking down. This phenomenon will also be studied in future work.

We can conclude from these tests that the new \texttt{PETSc}-based discrete gradient integrator is as accurate and as conservative as existing explicit symplectic methods in the \texttt{PETSc-PIC} framework. Additionally, when compared to previous implementations~\cite{Kraus2017_GEMPIC,Hu2024}, this new framework either matches or outperforms the accuracy, structure-preserving quality and computational efficiency of these previous implementations. As discussed, structure-preserving properties of the system can be improved when using discrete gradient integrators by reducing the relative tolerance of the L-BFGS solver, but this is likely limited by the solver's ability to continue making incremental improvements to the Jacobian matrix. Without a true Jacobian for the Vlasov-Poisson system derived, or a more advanced nonlinear solver, the L-BFGS solver paired with the discrete gradient integrator requires more time to converge than the symplectic methods, with only a small improvement in overall accuracy and structure preservation. Future work must, therefore, focus on introducing new nonlinear solvers to achieve lower convergence levels needed for accuracy and structure preservation. It is also worth noting that the implicit formulation of discrete gradient integrators also allows the time stepper to take larger stable time steps than the explicit methods.

\subsection{Collisional Dynamics}\label{sec:collisional}
To test the DGDI on collisional dynamics, we apply the particle basis Landau collision operator to a velocity space distribution of particles of two species in a test of equilibration. This tests the ability of the implementation to relax two different Maxwellian distributions to an equilibrium state. While this test is non-exhaustive, it provides verification of the structure-preserving qualities of the operator and integrator as well as the characteristic Maxwellian steady state of the operator. We also run the same test case using a simple Euler time integration technique, mimicking the methodology used by Carrillo in~\cite{Carrillo2020} to compare the structure preserving quality of each method. A more in depth discussion of the collision operator implementation in the \texttt{PETSc} library can be found in \cite{pusztay2023landaucollisionintegralparticle}, in preparation for submission at the time of writing, as well as \cite{Zonta2022, Hirvijoki2021, Carrillo2020}.
\begin{figure}[!h]
  \begin{center}
    \includegraphics[width=0.9\columnwidth]{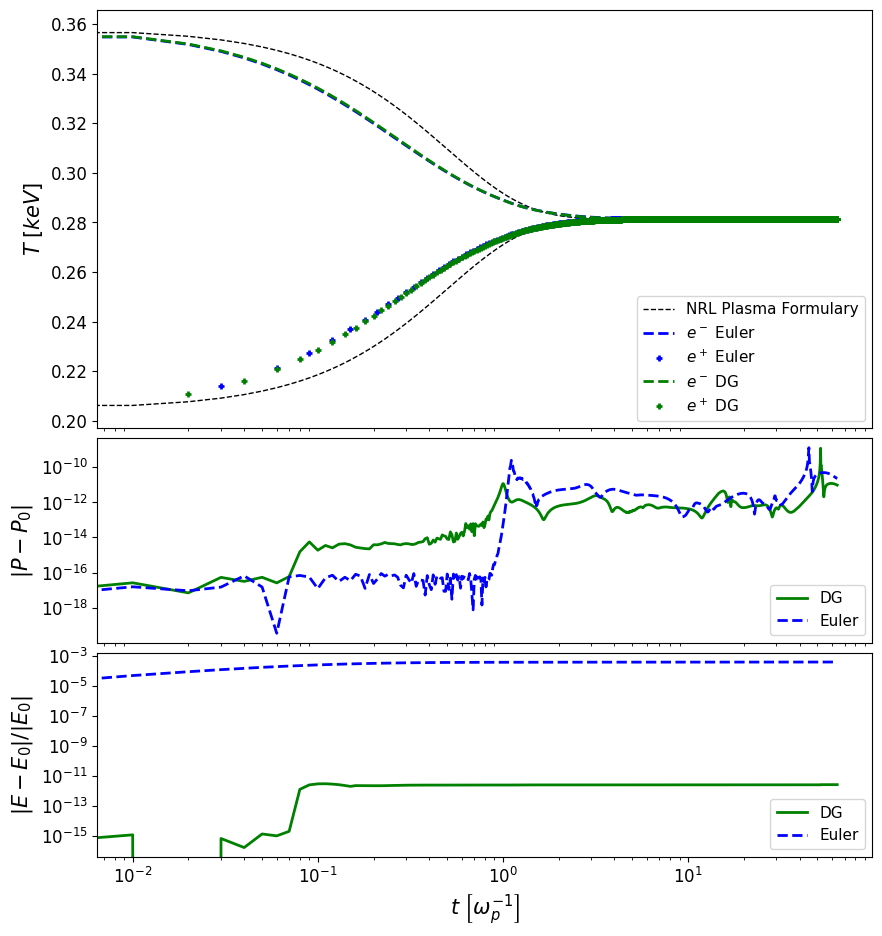}
  \end{center}
  \caption{(Top) Electron-positron equilibration using a non-conservative Euler and the average discrete gradient time integrator for the particle basis Landau collision integral compared with analytical solution. (Middle) Absolute error in momentum for electron-positron equilibration for each time integrator.(Bottom) Relative error in total energy for electron-positron equilibration for each time integrator.}
  \label{fig:equilibration}
\end{figure}
\begin{figure}[!h]
  \begin{center}
    \includegraphics[width=0.9\columnwidth]{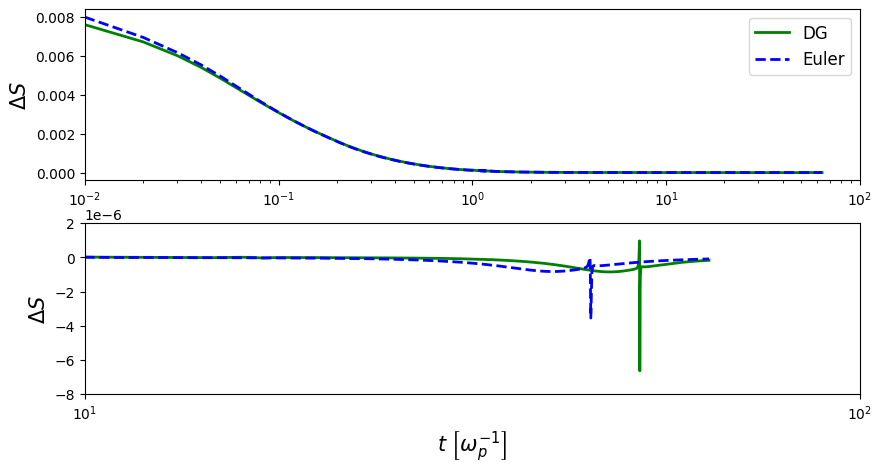}
  \end{center}
  \caption{(Top) Entropy change between time steps for the DGDI and Euler integrator. (Bottom) The same entropy change plot zoomed in on non-monontonic regions for each integrator.}
  \label{fig:equilibration_error}
\end{figure}

For simplicity, we choose electrons and positrons for our two species with initial temperatures of $0.35\left[KeV\right]$ and $0.2\left[KeV\right]$, respectively, and a simple mass ratio of $m_1/m_2=1$. As in the collisionless test, particles are laid out in a uniform grid, with weights sampled from the equilibrium Maxwellian at their respective species' temperatures. These initial Maxwellian weights are given by,
\begin{align}
    w_p = \int\int_{\boldsymbol{v}_p - h/2}^{\boldsymbol{v}_p + h/2} \frac{1}{\sqrt{2\pi\theta}} e^{(-\boldsymbol{v}^2/\theta)} \mathrm{d} \boldsymbol{v},
\end{align}
with $\theta = k_BT_s/m_s$, $h = 2L/N$ in the domain $[-L,L]^d$, and $N$ the number of particles per velocity space dimension. In this test, the domain size is $L=5 v_{th,s}$ where $v_{th,s}$ is the thermal velocity of the species, and the number of particles per species is $N=400$. Note that the spatial dimension, $\boldsymbol{x}$, has been omitted for this test as the particles exist in a single spatial cell, and collisions are applied locally. We use an initial time step size of $\Delta t = 0.01$, and allow the solver to reduce the size of the time step in situations where the nonlinear system becomes more difficult to solve.

We verify the results of this test using the thermal equilibration equations provided in the Naval Research Laboratory (NRL) Plasma Formulary~\cite{nrlformulary2023}. For two species, $\alpha$ and $\beta$, the thermal equilibration is given by,
\begin{equation}
    \frac{d T_{\alpha}}{dt} = \sum_{\beta} \bar{\nu}^{\alpha \backslash \beta} \left(T_{\beta} - T_{\alpha}\right),
\end{equation}
where $T_{\alpha}$ and $T_{\beta}$ are species temperatures and $\bar{\nu}^{\alpha \backslash \beta}$ is the interspecies collision rate. According to the Formulary, the collision rate can be written as,
\begin{equation}
    \bar{\nu}^{\alpha \backslash \beta} = 1.8\times 10^{-19} \frac{\left(m_{\alpha} m_{\beta}\right)^{1/2} Z_{\alpha}^{2} Z_{\beta}^{2} n_{\beta} \lambda_{\alpha \beta}}{\left(m_{\alpha} T_{\beta}+m_{\beta} T_{\alpha}\right)^{3/2}} \left[sec^{-1}\right],
\end{equation}
where
$m_{\alpha}$ and $m_{\beta}$ are species masses, $Z_{\alpha}$ and $Z_{\beta}$ are species charge states, $n_{\beta}$ the density of species $\beta$, and $\lambda_{\alpha \beta}$ the Coulomb logarithm. For our test with electrons and positrons, we normalize all the constants such that the collision rate becomes, $\bar{\nu}^{e \backslash p}=1$.

Shown in Figure~\ref{fig:equilibration}, both the DGDI and Euler integrators achieve full equilibration with some discrepancy between the data and the NRL equilibration rates. However, the equations presented in the NRL Plasma Formulary assume the plasma distributions remain exactly Maxwellian throughout the simulation. While the test thermalizes to a Maxwellian state, intermediate states may deviate from Maxwellian as it evolves and, therefore, some deviation from the NRL solution is expected. This is in line with previous works~\cite{Zonta2022,Adams2025_2}. Both time integration methods are relatively momentum conservative, with the average absolute error for the DGDI and Euler integrators being $\mathcal{O}(10^{-13})$ and $\mathcal{O}(10^{-11})$, respectively. We again use absolute error for momentum because the true values are close to zero. As noted by Carrillo in~\cite{Carrillo2020}, we expect momentum conservation to be unaffected by the use of a non-conservative time integrator up to $\mathcal{O} (\Delta t)$. The DGDI, however, achieves kinetic energy conservation on the order of $\mathcal{O}(10^{-12})$ while the Euler method does not, with relative kinetic energy errors of $\mathcal{O}(10^{-4})$. At early times, the discrete gradient momentum and kinetic energy error results are one to two orders of magnitude improved from the errors reported in~\cite{Hu2024} which uses a fixed point iterative solver. However, at around $t\sim0.8\left[\omega_p^{-1}\right]$ in the kinetic energy plot, the nonlinear solver reaches a configuration at which point it no longer achieves full convergence to the set tolerances. This explains the jump in energy at that time. The same phenomena happens around $t\sim1.0\left[\omega_p^{-1}\right]$ in the momentum plot but this jump does not appear in the energy data. 
These jumps do correlate, in part, to particle sampling where cases with a sparse sampling of the velocity space produce more instances where reaching solver tolerance is difficult. This coincides with the monotonicity of entropy as well, however, in tests increased sampling only alleviates this to a degree, with $N = 20$ being sufficient to mostly mitigate sampling related issues. The study of these solver failures is currently under investigation and will be explored in future work.

In Figure~\ref{fig:equilibration_error}, we show that using both the DGDI and Euler integrators results in monotonic regularized entropy of the system for much of the runtime. However, after the two species have been at equilibrium for a sufficiently long amount of time, both integrators encounter small non-monotonic fluctuations. The non-monotonic fluctuations are likely caused by small numerical errors, either introduced by the nonlinear solver with the DGDI or the non-conservative nature of the Euler method, where the entropy production of the system is almost exactly zero. Of particular interest is the monotonicity of entropy observed in the Euler data prior to equilibration. This is likely due to the inherent structure preserving nature of the particle-basis Landau operator, not easily broken by a poor choice in timestepping integrators. In future, more rigorous tests, we expect the DGDI method to outperform the non-conservative methods.

Thus, we have developed and tested a fully structure-preserving integration technique for the Landau equation that shows a clear advantages to previous methodologies~\cite{Carrillo2020}. Furthermore, in these collisional tests, it appears the use of an approximate nonlinear solver has had less impact on the structure-preserving qualities of the metric bracket than on the Poisson bracket, as the both the momentum and energy errors remains several orders of magnitude lower in this test than the collisionless tests. However, as with the collisionless tests, the DGDI was significantly less computationally efficient than the explicit methods. Careful considerations must therefore be made when choosing a time integration method for this implementation of the Landau operator. In future work, we plan to test this new integration method on more demanding test cases, such as two-species isotropization and Spitzer-resistivity.

\section{\label{sec:Conclusions}Conclusions}

We have presented the discrete gradient integrator, a structure-preserving time stepper for Hamiltonian systems, and considered the extension of this integrator to each component of the full metriplectic system: the collisionless Vlasov-Poisson system and the dissipative Landau equation for plasma collisions. Both the original Hamiltonian integrator and the collisional integrator, which we refer to as a \textit{discrete gradient dependent integrator}, were implemented in the \texttt{PETSc} library and tested on a suite of standard plasma tests: Landau damping, and electron-positron equilibration.

In the collisionless Landau damping test, we compared the discrete gradient integrator to the currently-used symplectic methods and a fourth-order Runge-Kutta method. The accuracy of the discrete gradients and symplectic integrators were comparable, while the Runge-Kutta method failed to capture the electric field damping and preserve any of the non-trivial conserved variables. The iterative nonlinear solver used for the implicit discrete gradient integrator method was significantly slower than the solver used for the explicit symplectic and Runge-Kutta methods, on average increasing the step time by $\sim 130\%$ and  $\sim 160\%$, respectively. This would seemingly indicate that the existing explicit symplectic timestepping integrators, which have similar structure-preserving qualities, are the preferred choice moving forward. However, by reducing the relative tolerance of the nonlinear solver, an improvement in the structure-preserving properties of the discrete gradient integrator method was found. While a tolerance limit was reached with the current L-BFGS method, an improved nonlinear solver paired with the discrete gradient framework should outperform the symplectic time stepper in structure-preserving quality. Generalized Broyden methods~\cite{Gerber1981} are a likely candidate for this and will be the focus of a future work. Furthermore, this conclusion does not take into account the larger time step sizes achievable using the implicit timestepping methods. Previous works have utilized implicit~\cite{Markidis2011}, or semi-implicit methods~\cite{Lapenta2017} to achieve ``exact'' energy conservation with larger time step sizes by making modifications to the PIC algorithm which could be implemented in the \texttt{PETSc-PIC} algorithm in future work. Optimizations to the nonlinear solvers would also make the implicit discrete gradient method a more competitive option in future work.

In the collisional, electron-positron equilibration test, the DGDI preserved the momentum to $\mathcal{O}(10^{-13})$ and kinetic energy to $\mathcal{O}(10^{-12})$. The entropy remained monotonic out to thermal equilibrium where some small non-monotonic fluctuations were encountered. As discussed, these non-monotonic fluctuations are likely caused by algebraic tolerances of the nonlinear solver and not a concern. This is a significant improvement from the original Euler time integrator presented in~\cite{Carrillo2020}. Along with the the work presented in~\cite{pusztay2023landaucollisionintegralparticle} and a separate implementation in~\cite{Zonta2022}, which uses a different machine architecture and methodology for computing integrals, this provides a cross-verification of the particle-basis Landau operator. Moreover, we have provided this new implementation in an open-source, community supported framework: \texttt{PETSc}. 

The ultimate goal of this work is to construct a combined metriplectic Vlasov-Poisson-Landau (or Vlasov-Maxwell-Landau) framework for plasma simulations. In simplest terms, the combined system can be evolved with~\eqref{eq:MetriplecticBracket}, summing the two brackets at each step. Combining the two brackets, however, would require significant computational costs, likely involving subcycling of one of the brackets at each step. Thus some efforts must first be put into improving the efficiency of the algorithm. A number of other improvements to the \texttt{PETSc-PIC} algorithm are currently in development, such as the inclusion of $H(div)$ finite elements, which preserve the electric field smoothness between cells, and noise reducing particle remapping algorithms. In~\cite{pusztay2023landaucollisionintegralparticle}, we refine aspects of the \texttt{PETSc} implementation, with additional test cases such as multi-species isotropization with realistic electron-ion mass ratios, omitted here for brevity.
In future work, we intend to derive a Jacobian matrix for the Vlasov-Poisson system and retest the discrete gradient integrator with a more efficient nonlinear solver. This is also applicable to the Landau operator system as well, however, deriving a Jacobian for this set of equations is a more complex task. Similar implementations in \texttt{PETSc}~\cite{Adams2017}, and elsewhere~\cite{Zonta2022}, have been ported to GPUs using CUDA, or a Kokkos~\cite{Kokkos2014} backend for portability requirements. Work to port this new collision algorithm has already begun leveraging the PetscDevice aspect of the library, and will provide the means for real-world physics studies in exascale environments.

\appendix

\section{Conservation Proofs for the Continuous Poisson Bracket} \label{sec:AppendixA}

\subsection{Note on Divergence Theorem for Periodic Boxes}

On a fully periodic domain $\left[0,L\right]^3$, any smooth function, $F(\boldsymbol{x})$, satisfies,
\begin{equation}
  \int_{\left[0,L\right]^3} \nabla_{\boldsymbol{x}} F(\boldsymbol{x}) \mathrm{d} \boldsymbol{x} = 0,
\end{equation}
\cite{Heath2012,Perin2014}. Equivalently, if you view $\left[0,L\right]^3$ as a box whose opposite faces are identified, then writing out the divergence theorem in with dimensional indices,
\begin{equation*}
  \int_{\left[0,L\right]^3} \frac{\partial}{\partial x_i} F(\boldsymbol{x}) \mathrm{d} \boldsymbol{x} = \int_{\left[0,L\right]^3_{(x_j,x_k)}} \left[F(x_i=L,x_j,x_k) - F(x_i=0,x_j,x_k)\right] \mathrm{d} x_j \mathrm{d} x_k
\end{equation*}
where $\left\{i,j,k\right\}=\left\{1,2,3\right\}$. However, periodicity means, 
\begin{equation*}
  F(x_i=L,x_j,x_k) = F(x_i=0,x_j,x_k)\;\;\text{for all }(x_j,x_k),
\end{equation*}
so the integrand on each face is zero. Hence
\begin{equation}
  \int_{\left[0,L\right]^3} \frac{\partial}{\partial x_i} F(\boldsymbol{x}) \mathrm{d} \boldsymbol{x} = 0.
\end{equation}
This property also holds if the spatial variable, $\boldsymbol{x}$, is swapped for velocity, $\boldsymbol{v}$.

\subsection{Continuous Equations}
As a reminder, the written out continuous form of the Poisson bracket is,
\begin{align*}
  \left\{\mathcal{F},\mathcal{G}\right\} \left[f_s\right] = &\sum_s \int \frac{f_s}{m_s} \left[\nabla_{\boldsymbol{x}} \frac{\delta \mathcal{F}}{\delta f_s} \cdot \nabla_{\boldsymbol{v}} \frac{\delta \mathcal{G}}{\delta f_s} - \nabla_{\boldsymbol{x}} \frac{\delta \mathcal{G}}{\delta f_s} \cdot \nabla_{\boldsymbol{v}} \frac{\delta \mathcal{F}}{\delta f_s} \right] d\boldsymbol{z},
\end{align*}
and the mass, momentum and Hamiltonian functionals, plus entropy and regularized entropy, are given by,
\begin{align*}
    \mathcal{M} &= \sum_s \int  m_s  f_s \mathrm{d} \boldsymbol{z},\\
    \mathcal{P} &= \sum_s \int  m_s |\boldsymbol{v}| f_s \mathrm{d} \boldsymbol{z},\\
    \mathcal{H} &= \sum_s \int \frac{1}{2} m_s |\boldsymbol{v}|^2 f_s \mathrm{d} \boldsymbol{z} + \frac{1}{2} \int |\nabla \phi|^2 \mathrm{d} \boldsymbol{z}, \\
    \mathcal{S} &= - \sum_s \int f_s \ln f_s \mathrm{d} \boldsymbol{z},\\
    \mathcal{S}_\epsilon &=-\sum_s \int f_s \ast \psi_{\epsilon}\left(\boldsymbol{z}\right) \ln \left(f_s \ast \psi_{\epsilon}\left(\boldsymbol{z}\right)\right) \mathrm{d} \boldsymbol{z}. 
\end{align*}
We will now show each of these functionals is conserved by the Poisson bracket.

\subsection{Mass}
Conservation of mass is simple to prove as,
\begin{equation*}
    \frac{\delta \mathcal{M}}{\delta f_s} = m_s \rightarrow \nabla_{\boldsymbol{x}} \frac{\delta \mathcal{M}}{\delta f_s} = 0, \;\;\nabla_{\boldsymbol{v}} \frac{\delta \mathcal{M}}{\delta f_s} = 0,
\end{equation*}
and therefore $\left\{\mathcal{M},\mathcal{G}\right\} = 0$. 

\subsection{Momentum}

Recalling the definition of the momentum functional~\eqref{eq:MomFunc}, the functional derivative with respect to $f_s$ is,
\begin{equation*}
  \frac{\delta \mathcal{P}}{\delta f_s} = m_s \boldsymbol{v},
\end{equation*}
and the gradients of the functional derivative $\frac{\delta \mathcal{P}}{\delta f_s}$ are,
\begin{equation*}
  \nabla_{\boldsymbol{x}} \frac{\delta \mathcal{P}}{\delta f_s} = 0, \;\; \nabla_{\boldsymbol{v}} \frac{\delta \mathcal{P}}{\delta f_s} = m_s \mathbb{I},
\end{equation*}
Plugging the momentum functional into the Poisson bracket~\eqref{eq:PoissonBracket} with the above relations then gives,
\begin{align*}
  \left\{\mathcal{P},\mathcal{G}\right\} \left[f_s\right] = &\sum_s \int \frac{f_s}{m_s} \left[\nabla_{\boldsymbol{x}} \frac{\delta \mathcal{P}}{\delta f_s} \cdot \nabla_{\boldsymbol{v}} \frac{\delta \mathcal{G}}{\delta f_s} - \nabla_{\boldsymbol{x}} \frac{\delta \mathcal{G}}{\delta f_s} \cdot \nabla_{\boldsymbol{v}} \frac{\delta \mathcal{P}}{\delta f_s} \right] d\boldsymbol{z} \\
  &= -\sum_s \int f_s \left[\nabla_{\boldsymbol{x}} \frac{\delta \mathcal{G}}{\delta f_s} \right] d\boldsymbol{z},
\end{align*}
This integral does not obviously vanish, as $\mathcal{G}$ may be a function of $\boldsymbol{x}$.
Thus, we cannot prove that the momentum functional is a Casimir of the continuous bracket. However, if we substitute the Hamiltonian functional, $\mathcal{H}$ into the bracket for $\mathcal{G}$, we may show that momentum is a conserved moment of the system. First consider the functional derivatives of the Hamiltonian,
\begin{equation*}
  \frac{\delta \mathcal{H}}{\delta f_s} = \frac{m_s}{2} |\boldsymbol{v}|^2,
\end{equation*}
and its gradients,
\begin{equation*}
  \nabla_{\boldsymbol{x}} \frac{\delta \mathcal{H}}{\delta f_s} = 0,\;\;\nabla_{\boldsymbol{v}} \frac{\delta \mathcal{H}}{\delta f_s} = m_s \boldsymbol{v}
\end{equation*}
Returning to the reduced form of the Poisson bracket with the momentum functional we can plug the Hamiltonian functional in to get,
\begin{align*}
  \left\{\mathcal{P},\mathcal{H}\right\} &= -\sum_s \int f_s \left[\nabla_{\boldsymbol{x}} \frac{\delta \mathcal{H}}{\delta f_s} \right] d\boldsymbol{z}\\
  &= 0,\\
\end{align*}
and momentum is conserved.

\subsection{Hamiltonian (Total Energy)}

Proving the Hamiltonian functional commutes with itself is straightforward for the continuous Poisson bracket as the bracket is antisymmetric. Plugging, $\mathcal{H}$ into the Poisson bracket with itself reveals,
\begin{align*}
  \left\{\mathcal{H},\mathcal{H}\right\} \left[f_s\right] = &\sum_s \int \frac{f_s}{m_s} \left[\nabla_{\boldsymbol{x}} \frac{\delta \mathcal{H}}{\delta f_s} \cdot \nabla_{\boldsymbol{v}} \frac{\delta \mathcal{H}}{\delta f_s} - \nabla_{\boldsymbol{x}} \frac{\delta \mathcal{H}}{\delta f_s} \cdot \nabla_{\boldsymbol{v}} \frac{\delta \mathcal{H}}{\delta f_s} \right] d\boldsymbol{z},\\
  &= 0.
\end{align*}

\subsection{Entropy}

Recalling the definition of the entropy functional~\eqref{eq:EntFunc}, the functional derivative with respect to $f_s$ is,
\begin{equation*}
  \frac{\delta \mathcal{S}}{\delta f_s} = - \left[\log (f_s) + 1\right],
\end{equation*}
and the gradients of the functional derivative $\frac{\delta \mathcal{S}}{\delta f_s}$ are,
\begin{equation*}
  \nabla_{\boldsymbol{x}} \frac{\delta \mathcal{S}}{\delta f_s} = -\frac{1}{f_s} \nabla_{\boldsymbol{x}} f_s, \;\; \nabla_{\boldsymbol{v}} \frac{\delta \mathcal{S}}{\delta f_s} = -\frac{1}{f_s} \nabla_{\boldsymbol{v}} f_s,
\end{equation*}
Plugging these relations into the Poisson bracket reveals,
\begin{align*}
  \left\{\mathcal{S},\mathcal{G}\right\} &= \sum_s \int \frac{f_s}{m_s} \left[\nabla_{\boldsymbol{x}} \frac{\delta \mathcal{S}}{\delta f_s} \cdot \nabla_{\boldsymbol{v}} \frac{\delta \mathcal{G}}{\delta f_s} - \nabla_{\boldsymbol{x}} \frac{\delta \mathcal{G}}{\delta f_s} \cdot \nabla_{\boldsymbol{v}} \frac{\delta \mathcal{S}}{\delta f_s} \right] d\boldsymbol{z}, \\
  &= -\sum_s \int \frac{1}{m_s} \left[\left( \nabla_{\boldsymbol{x}} f_s\right) \cdot \nabla_{\boldsymbol{v}} \frac{\delta \mathcal{G}}{\delta f_s} - \nabla_{\boldsymbol{x}} \frac{\delta \mathcal{G}}{\delta f_s} \cdot \left(\nabla_{\boldsymbol{v}} f_s\right) \right] d\boldsymbol{z}.
\end{align*}
If we integrate each component of this integral by parts, ignoring boundary terms again, we get,
\begin{align*}
  \left\{\mathcal{S},\mathcal{G}\right\} &= -\sum_s \int \frac{f_s}{m_s} \left[-\nabla_{\boldsymbol{v}} \cdot \left( \nabla_{\boldsymbol{x}} \frac{\delta \mathcal{G}}{\delta f_s}\right) + \nabla_{\boldsymbol{x}} \cdot \left(\nabla_{\boldsymbol{v}} \frac{\delta \mathcal{G}}{\delta f_s}\right) \right] d\boldsymbol{z},\\
  &= \sum_s \int \frac{1}{m_s} \left[\nabla_{\boldsymbol{x}} \cdot \left( \nabla_{\boldsymbol{v}} \frac{\delta \mathcal{G}}{\delta f_s}\right) - \nabla_{\boldsymbol{x}} \cdot \left(\nabla_{\boldsymbol{v}} \frac{\delta \mathcal{G}}{\delta f_s}\right) \right] d\boldsymbol{z},\\
  &= 0,
\end{align*}
since the gradients commute. Thus, entropy is also a Casimir of the continuous bracket.

\subsection{Regularized Entropy}

The functional derivative of the regularized entropy~\eqref{RegEntropy} is given as,
\begin{align*}
    \frac{\delta \mathcal{S}_{\epsilon}}{\delta f_s} &= -\psi_{\epsilon} \ast \left(1+\ln \left(f_s \ast \psi_{\epsilon}\right)\right),
\end{align*}
which we will define as,
\begin{align*}
    \frac{\delta \mathcal{S}_{\epsilon}}{\delta f_s} &= - \Phi_s (\boldsymbol{z}).
\end{align*}
For simplicity, we can for now avoid showing the gradients of this functional derivative and write the Poisson bracket as,
\begin{align*}
  \left\{\mathcal{S}_{\epsilon},\mathcal{G}\right\} \left[f_s\right] &= \sum_s \int \frac{f_s}{m_s} \left[\nabla_{\boldsymbol{x}} \frac{\delta \mathcal{S}_{\epsilon}}{\delta f_s} \cdot \nabla_{\boldsymbol{v}} \frac{\delta \mathcal{G}}{\delta f_s} - \nabla_{\boldsymbol{x}} \frac{\delta \mathcal{G}}{\delta f_s} \cdot \nabla_{\boldsymbol{v}} \frac{\delta \mathcal{F}}{\delta f_s} \right] \mathrm{d}\boldsymbol{z}, \\
  &= -\sum_s \int \frac{f_s}{m_s} \left[\nabla_{\boldsymbol{x}} (\Phi_s (\boldsymbol{z})) \cdot \nabla_{\boldsymbol{v}} \frac{\delta \mathcal{G}}{\delta f_s} - \nabla_{\boldsymbol{x}} \frac{\delta \mathcal{G}}{\delta f_s} \cdot \nabla_{\boldsymbol{v}} (\Phi_s (\boldsymbol{z})) \right] \mathrm{d}\boldsymbol{z}. \\
\end{align*}
As with the momentum proof, we can show that this bracket vanishes when the Hamiltonian functional, $\mathcal{H}$, is plugged in,
\begin{align*}
  \left\{\mathcal{S}_{\epsilon},\mathcal{H}\right\} &= -\sum_s \int \frac{f_s}{m_s} \left[\nabla_{\boldsymbol{x}} (\Phi_s (\boldsymbol{z})) \cdot \nabla_{\boldsymbol{v}} \frac{\delta \mathcal{H}}{\delta f_s} - \nabla_{\boldsymbol{x}} \frac{\delta \mathcal{H}}{\delta f_s} \cdot \nabla_{\boldsymbol{v}} (\Phi_s (\boldsymbol{z})) \right] \mathrm{d} \boldsymbol{z}, \\
  &= -\sum_s \int f_s \left[\nabla_{\boldsymbol{x}} (\Phi_s (\boldsymbol{z})) \cdot \boldsymbol{v} \right] \mathrm{d} \boldsymbol{z}. \\
\end{align*}
We can write $(\boldsymbol{v} \cdot \nabla_{\boldsymbol{x}}) \Phi_s = \nabla_{\boldsymbol{x}} (\boldsymbol{v} \Phi_s) - \Phi_s (\nabla_{\boldsymbol{x}} \cdot \boldsymbol{v})= \nabla_{\boldsymbol{x}} (\boldsymbol{v} \Phi_s)$, and therefore $ f_s \nabla_{\boldsymbol{x}} \cdot (\boldsymbol{v} \Phi_s) = \nabla_{\boldsymbol{x}} \cdot (f_s  \boldsymbol{v} \Phi_s) - (\nabla_{\boldsymbol{x}} f_s) \cdot (\boldsymbol{v} \Phi_s)$. The first integral then becomes, 
\begin{equation*}
    \left\{\mathcal{S}_{\epsilon},\mathcal{H}\right\}  = -\sum_s \int \nabla_{\boldsymbol{x}} \cdot (f_s  \boldsymbol{v} \Phi_s) \mathrm{d} \boldsymbol{z} + \sum_s \int (\nabla_{\boldsymbol{x}} f_s) \cdot (\boldsymbol{v} \Phi_s) \mathrm{d} \boldsymbol{z}.
\end{equation*}
The first term is a divergence in $\boldsymbol{x}$, which vanishes under periodic boundary conditions. Using the same product rule trick on the second term we can show,
\begin{equation*}
    \left\{\mathcal{S}_{\epsilon},\mathcal{H}\right\} = \sum_s \int f_s \boldsymbol{v} \cdot (\nabla_{\boldsymbol{x}} \Phi_s) \mathrm{d} \boldsymbol{z}.
\end{equation*}
which is the negative of the original bracket before the expansion. Therefore, the second term must be zero and the regularized entropy is a conserved quantity of the continuous bracket.

\section{Running Discrete Gradient Tests in \texttt{PETSc} Library}

The data presented in this paper can be recreated with \texttt{PETSc} using the \texttt{TS} example ex2 (\texttt{\$PETSC\_DIR/src/ts/tutorials/hamiltonian/ex2.c}). Exact runtimes may vary depending on the architecture and compiler. The \texttt{TS} example can be run using the following options:
\begin{verbatim}
./ex2 -dm_plex_dim 1 -dm_plex_simplex 0 -dm_plex_box_faces 80 \
 -dm_plex_box_lower 0. -dm_plex_box_upper 12.5664 \
 -dm_plex_box_bd periodic -dm_swarm_num_species 1 \
 -vdm_plex_dim 1 -vdm_plex_simplex 0 -vdm_plex_box_faces 6000 \
 -vdm_plex_box_lower -10 -vdm_plex_box_upper 10 \
 -charges -1.,1. -sigma 1.0e-8  -petscspace_degree 1 \
 -cosine_coefficients 0.01,0.5 -perturbed_weights -total_weight 1. \
 -ts_dt 0.01 -ts_max_time 30 -ts_max_steps 3000 -remap_freq 0\
 -em_snes_atol 1.e-15 -em_type primal -em_pc_type svd \
 -ts_type discgrad -ts_discgrad_type average \
 -snes_fd -snes_type qn -snes_qn_type lbfgs 
\end{verbatim}
To run either the symplectic or Runge-Kutta case, the user can replace the options \texttt{-ts\_type discgrad -ts\_discgrad\_type average} with \texttt{-ts\_type basicsymplectic -ts\_basicsymplectic\_type 1}
or \texttt{-ts\_type rk -ts\_rk\_type 4}, respectively. This change will also make the final line, specifying the nonlinear solver type, unnecessary, as the new explicit method will use a linear Krylov solver.

\section*{Acknowledgments}

This project was supported by the Office of Naval Research and by an appointment to the NRC Research Associateship Program at the Naval Research Laboratory, administered by the Fellowships Office of the National Academies of Sciences, Engineering, and Medicine. This material is also based upon work supported by the U.S. Department of Energy, Office of Science, Office of Advanced Scientific Computing Research and Office of Fusion Energy Sciences, Scientific Discovery through Advanced Computing (SciDAC) program through the FASTMath Institute under Contract No. DE-AC02-05CH11231 at Lawrence Berkeley National Laboratory. DSF and MGK were partially supported by NSF CSSI grant 1931524.

\appendix

\bibliographystyle{elsarticle-num}
\bibliography{ref_DG}

@article{Adams2017,
author = {Adams, Mark F. and Hirvijoki, Eero and Knepley, Matthew G. and Brown, Jed and Isaac, Tobin and Mills, Richard},
title = {Landau Collision Integral Solver with Adaptive Mesh Refinement on Emerging Architectures},
journal = {SIAM Journal on Scientific Computing},
volume = {39},
number = {6},
pages = {C452-C465},
year = {2017},
doi = {10.1137/17M1118828},
URL = {https://doi.org/10.1137/17M1118828},
eprint = {https://doi.org/10.1137/17M1118828}
}

@article{Adams2025,
      title={A projection method for particle resampling}, 
      author={Mark F. Adams and Daniel S. Finn and Matthew G. Knepley and Joseph V. Pusztay},
      year={2025},
      eprint={2501.13681},
      archivePrefix={arXiv},
      primaryClass={physics.plasm-ph},
      url={https://arxiv.org/abs/2501.13681}, 
}

@article{Adams2025_2,
	author = {Adams, Mark F. and Wang, Peng and Merson, Jacob and Huck, Kevin and Knepley, Matthew G.},
	doi = {10.1137/24M1640252},
	eprint = {https://doi.org/10.1137/24M1640252},
	journal = {SIAM Journal on Scientific Computing},
	number = {2},
	pages = {B360-B381},
	title = {A Performance Portable, Fully Implicit Landau Collision Operator with Batched Linear Solvers},
	url = {https://doi.org/10.1137/24M1640252},
	volume = {47},
	year = {2025}}

@techreport{AbhyankarBrownConstantinescuGhoshSmith2014,
  title       = {{PETSc/TS}: A Modern Scalable {DAE/ODE} Solver Library},
  author      = {Shrirang Abhyankar and Jed Brown and Emil Constantinescu and Debojyoti Ghosh and Barry F. Smith},
  type        = {Preprint},
  number      = {ANL/MCS-P5061-0114},
  institution = {ANL},
  month       = {1},
  year        = {2014},
  petsc_uses  = {TS}
}

@book{BirdsallLangdon,
	author = {Birdsall, Charles K. and Langdon, A. Bruce.},
	title = {Plasma physics via computer simulation},
	publisher = {McGraw-Hill},
	year = {c1985}
}

@article{BrownKnepleySmith14,
	author = {Jed Brown and Matthew G. Knepley and Barry Smith},
	doi = {10.1109/MCSE.2014.95},
	journal = {IEEE Computing in Science and Engineering},
	month = {1},
	number = {1},
	pages = {38--45},
	petsc_uses = {KSP},
	title = {Run-time extensibility and librarization of simulation software},
	volume = {17},
	year = {2015}
}

@article{Broyden1970,
    author = {Broyden, C. G.},
    title = {The Convergence of a Class of Double-rank Minimization Algorithms 1. General Considerations},
    journal = {IMA Journal of Applied Mathematics},
    volume = {6},
    number = {1},
    pages = {76-90},
    year = {1970},
    month = {03},
    issn = {0272-4960},
    doi = {10.1093/imamat/6.1.76},
    url = {https://doi.org/10.1093/imamat/6.1.76},
    eprint = {https://academic.oup.com/imamat/article-pdf/6/1/76/2233756/6-1-76.pdf},
}

@article{Carrillo2020,
   author = {Jose A. Carrillo and Jingwei Hu and Li Wang and Jeremy Wu},
   doi = {10.1016/j.jcpx.2020.100066},
   issn = {25900552},
   journal = {Journal of Computational Physics: X},
   title = {A particle method for the homogeneous {Landau} equation},
   volume = {7},
   year = {2020},
}

@article{Cohen2011,
	author = {Cohen, David and Hairer, Ernst},
	date = {2011/03/01},
	doi = {10.1007/s10543-011-0310-z},
	id = {Cohen2011},
	isbn = {1572-9125},
	journal = {BIT Numerical Mathematics},
	number = {1},
	pages = {91--101},
	title = {Linear energy-preserving integrators for Poisson systems},
	url = {https://doi.org/10.1007/s10543-011-0310-z},
	volume = {51},
	year = {2011}}

@article{Chen2011,
title = {An energy- and charge-conserving, implicit, electrostatic particle-in-cell algorithm},
journal = {Journal of Computational Physics},
volume = {230},
number = {18},
pages = {7018-7036},
year = {2011},
issn = {0021-9991},
doi = {https://doi.org/10.1016/j.jcp.2011.05.031},
url = {https://www.sciencedirect.com/science/article/pii/S0021999111003421},
author = {G. Chen and L. Chacón and D.C. Barnes},
}

@article{Eero2020,
doi = {10.1088/1361-6587/abe884},
url = {https://doi.org/10.1088/1361-6587/abe884},
year = 2021,
month = {3},
publisher = {{IOP} Publishing},
volume = {63},
number = {4},
pages = {044003},
author = {Eero Hirvijoki},
title = {Structure-preserving marker-particle discretizations of {Coulomb} collisions for particle-in-cell codes},
journal = {Plasma Physics and Controlled Fusion}
}

@article{Esirkepov2001,
	author = {T.Zh. Esirkepov},
	doi = {https://doi.org/10.1016/S0010-4655(00)00228-9},
	issn = {0010-4655},
	journal = {Computer Physics Communications},
	number = {2},
	pages = {144-153},
	title = {Exact charge conservation scheme for Particle-in-Cell simulation with an arbitrary form-factor},
	url = {https://www.sciencedirect.com/science/article/pii/S0010465500002289},
	volume = {135},
	year = {2001}
}

@article{Evstatiev2013,
   author = {E. G. Evstatiev and B. A. Shadwick},
   doi = {10.1016/j.jcp.2013.03.006},
   issn = {10902716},
   journal = {Journal of Computational Physics},
   title = {Variational formulation of particle algorithms for kinetic plasma simulations},
   volume = {245},
   year = {2013},
}

@article{Finn2023, 
    title={A numerical study of {Landau} damping with {PETSc-PIC}}, 
    volume={18}, 
    url={http://dx.doi.org/10.2140/camcos.2023.18.135}, 
    DOI={10.2140/camcos.2023.18.135}, 
    number={1}, 
    journal={Communications in Applied Mathematics and Computational Science}, 
    publisher={Mathematical Sciences Publishers}, 
    author={Finn, Daniel S. and Knepley, Matthew G. and Pusztay, Joseph V. and Adams, Mark F.}, 
    year={2023}, 
    month={12}, 
    pages={135–152}
}

@article{Gerber1981,
	author = {Gerber, Richard R. and Luk, Franklin T.},
	doi = {10.1137/0718061},
	eprint = {https://doi.org/10.1137/0718061},
	journal = {SIAM Journal on Numerical Analysis},
	number = {5},
	pages = {882-890},
	title = {A Generalized Broyden's Method for Solving Simultaneous Linear Equations},
	url = {https://doi.org/10.1137/0718061},
	volume = {18},
	year = {1981}}

@article{Gonzalez1996,
   author = {O. Gonzalez},
   doi = {10.1007/BF02440162},
   issn = {09388974},
   issue = {5},
   journal = {Journal of Nonlinear Science},
   title = {Time Integration and Discrete {Hamiltonian} Systems},
   volume = {6},
   year = {1996},
}

@Inbook{Harten1997,
author="Harten, Amiram and Lax, Peter D. and van Leer, Bram",
editor="Hussaini, M. Yousuff and van Leer, Bram and Van Rosendale, John",
title="On Upstream Differencing and Godunov-Type Schemes for Hyperbolic Conservation Laws",
bookTitle="Upwind and High-Resolution Schemes",
year="1997",
publisher="Springer Berlin Heidelberg",
address="Berlin, Heidelberg",
pages="53--79",
isbn="978-3-642-60543-7",
doi="10.1007/978-3-642-60543-7_4",
url="https://doi.org/10.1007/978-3-642-60543-7_4"
}

@article{Heath2012,
	author = {R.E. Heath and I.M. Gamba and P.J. Morrison and C. Michler},
	doi = {https://doi.org/10.1016/j.jcp.2011.09.020},
	issn = {0021-9991},
	journal = {Journal of Computational Physics},
	number = {4},
	pages = {1140-1174},
	title = {A discontinuous Galerkin method for the Vlasov--Poisson system},
	url = {https://www.sciencedirect.com/science/article/pii/S0021999111005523},
	volume = {231},
	year = {2012}}

@misc{Hirvijoki2018,
      title={Metriplectic Particle-in-Cell Integrators for the {Landau} Collision Operator}, 
      author={Eero Hirvijoki and Michael Kraus and Joshua W. Burby},
      year={2018},
      eprint={1802.05263},
      archivePrefix={arXiv},
      primaryClass={physics.comp-ph}
}

@article{Hirvijoki2020,
   author = {Eero Hirvijoki and Joshua W. Burby},
   doi = {10.1063/5.0011297},
   issn = {10897674},
   issue = {8},
   journal = {Physics of Plasmas},
   title = {Collisional gyrokinetics teases the existence of metriplectic reduction},
   volume = {27},
   year = {2020},
}

@article{Hirvijoki2021,
   author = {Eero Hirvijoki},
   doi = {10.1088/1361-6587/abe884},
   issn = {13616587},
   issue = {4},
   journal = {Plasma Physics and Controlled Fusion},
   title = {Structure-preserving marker-particle discretizations of {Coulomb} collisions for particle-in-cell codes},
   volume = {63},
   year = {2021},
}

@article{He2015,
   author = {Yang He and Hong Qin and Yajuan Sun and Jianyuan Xiao and Ruili Zhang and Jian Liu},
   doi = {10.1063/1.4938034},
   issn = {10897674},
   issue = {12},
   journal = {Physics of Plasmas},
   title = {{Hamiltonian} time integrators for {Vlasov-Maxwell} equations},
   volume = {22},
   year = {2015},
}

@article{Hu2024,
      title={Fully discrete energy-dissipative and conservative discrete gradient particle methods for a class of continuity equations}, 
      author={Jingwei Hu and Samuel Q. Van Fleet and Andy T. S. Wan},
      year={2024},
      eprint={2407.00533},
      archivePrefix={arXiv},
      primaryClass={math.NA},
      url={https://arxiv.org/abs/2407.00533}, 
}

@article{Itoh1988,
   author = {Toshiaki Itoh and Kanji Abe},
   doi = {10.1016/0021-9991(88)90132-5},
   issn = {10902716},
   issue = {1},
   journal = {Journal of Computational Physics},
   title = {{Hamiltonian}-conserving discrete canonical equations based on variational difference quotients},
   volume = {76},
   year = {1988},
}

@article{Kaufman1984,
   author = {Allan N. Kaufman},
   doi = {10.1016/0375-9601(84)90634-0},
   issn = {03759601},
   issue = {8},
   journal = {Physics Letters A},
   title = {Dissipative {Hamiltonian} systems: A unifying principle},
   volume = {100},
   year = {1984},
}

@article{Kokkos2014,
	author = {H. {Carter Edwards} and Christian R. Trott and Daniel Sunderland},
	doi = {https://doi.org/10.1016/j.jpdc.2014.07.003},
	issn = {0743-7315},
	journal = {Journal of Parallel and Distributed Computing},
	note = {Domain-Specific Languages and High-Level Frameworks for High-Performance Computing},
	number = {12},
	pages = {3202-3216},
	title = {Kokkos: Enabling manycore performance portability through polymorphic memory access patterns},
	url = {https://www.sciencedirect.com/science/article/pii/S0743731514001257},
	volume = {74},
	year = {2014}}

@article{Kraus2017,
   author = {Michael Kraus and Eero Hirvijoki},
   doi = {10.1063/1.4998610},
   issn = {10897674},
   issue = {10},
   journal = {Physics of Plasmas},
   title = {Metriplectic Integrators for the {Landau} Collision Operator},
   volume = {24},
   year = {2017},
}

@article{Kraus2017_GEMPIC,
   author = {Michael Kraus and Katharina Kormann and Philip J. Morrison and Eric Sonnendrücker},
   doi = {10.1017/S002237781700040X},
   issn = {14697807},
   issue = {4},
   journal = {Journal of Plasma Physics},
   title = {{GEMPIC}: {Geometric} electromagnetic particle-in-cell methods},
   volume = {83},
   year = {2017},
}

@article{Lapenta2017,
	author = {Giovanni Lapenta},
	doi = {https://doi.org/10.1016/j.jcp.2017.01.002},
	issn = {0021-9991},
	journal = {Journal of Computational Physics},
	pages = {349-366},
	title = {Exactly energy conserving semi-implicit particle in cell formulation},
	url = {https://www.sciencedirect.com/science/article/pii/S0021999117300128},
	volume = {334},
	year = {2017}
}

@article{Landau1936,
   author="L.D. Landau",
   title="Die kinetische Gleichung für den Fall {Coulombscher} Wechselwirkung",
   Journal="Phys. Zs. Sowjetunion",
   volume="10",
   pages="154–164",
   year="1936"
}

@article{Li2023,
	author = {Yingzhe Li},
	doi = {https://doi.org/10.1016/j.jcp.2022.111733},
	issn = {0021-9991},
	journal = {Journal of Computational Physics},
	pages = {111733},
	title = {Energy conserving particle-in-cell methods for relativistic {Vlasov--Maxwell} equations of laser-plasma interaction},
	url = {https://www.sciencedirect.com/science/article/pii/S0021999122007963},
	volume = {473},
	year = {2023}
}

@article{Liu2023,
	author = {Ju Liu},
	doi = {https://doi.org/10.1016/j.jcp.2023.112177},
	issn = {0021-9991},
	journal = {Journal of Computational Physics},
	pages = {112177},
	title = {On the design of non-singular, energy-momentum consistent integrators for nonlinear dynamics using energy splitting and perturbation techniques},
	url = {https://www.sciencedirect.com/science/article/pii/S0021999123002723},
	volume = {487},
	year = {2023}
   }

@article{Mansfield2009,
   author = {Elizabeth L Mansfield},
   journal = {Victoria},
   title = {On the construction of discrete gradients},
   year = {2009},
}

@article{Markidis2011,
    author = {Stefano Markidis and Giovanni Lapenta},
    doi = {https://doi.org/10.1016/j.jcp.2011.05.033},
    issn = {0021-9991},
    journal = {Journal of Computational Physics},
    number = {18},
    pages = {7037-7052},
    title = {The energy conserving particle-in-cell method},
    url = {https://www.sciencedirect.com/science/article/pii/S0021999111003445},
    volume = {230},
    year = {2011}
}

@article{Marsden1982,
   author = {Jerrold E. Marsden and Alan Weinstein},
   doi = {10.1016/0167-2789(82)90043-4},
   issn = {01672789},
   issue = {3},
   journal = {Physica D: Nonlinear Phenomena},
   title = {The {Hamiltonian} structure of the {Maxwell-Vlasov} equations},
   volume = {4},
   year = {1982},
}

@article{McLachlan1999,
   author = {Robert I. McLachlan and G. R.W. Quispel and Nicolas Robidoux},
   doi = {10.1098/rsta.1999.0363},
   issn = {1364503X},
   issue = {1754},
   journal = {Philosophical Transactions of the Royal Society A: Mathematical, Physical and Engineering Sciences},
   title = {Geometric integration using discrete gradients},
   volume = {357},
   year = {1999},
}

@article{McLaren2004,
   author = {D. I. McLaren and G. R.W. Quispel},
   doi = {10.1088/0305-4470/37/39/L01},
   issn = {03054470},
   issue = {39},
   journal = {Journal of Physics A: Mathematical and General},
   title = {Integral-preserving integrators},
   volume = {37},
   year = {2004},
}

@ARTICLE{Miloshevich2021,
       author = {{Miloshevich}, George and {Burby}, Joshua W.},
        title = "{Hamiltonian reduction of Vlasov-Maxwell to a dark slow manifold}",
      journal = {Journal of Plasma Physics},
         year = 2021,
        month = jun,
       volume = {87},
       number = {3},
          eid = {835870301},
        pages = {835870301},
          doi = {10.1017/S0022377821000556},
archivePrefix = {arXiv},
       eprint = {2104.03092},
 primaryClass = {physics.plasm-ph},
       adsurl = {https://ui.adsabs.harvard.edu/abs/2021JPlPh..87c8301M}
}

@ARTICLE{Mitchell2019,
       author = {{Mitchell}, Matthew S. and {Miecnikowski}, Matthew T. and {Beylkin}, Gregory and {Parker}, Scott E.},
        title = "{Efficient Fourier basis particle simulation}",
      journal = {Journal of Computational Physics},
         year = 2019,
        month = nov,
       volume = {396},
        pages = {837-847},
          doi = {10.1016/j.jcp.2019.07.023},
archivePrefix = {arXiv},
       eprint = {1808.03742},
 primaryClass = {physics.plasm-ph},
       adsurl = {https://ui.adsabs.harvard.edu/abs/2019JCoPh.396..837M}
}

@article{Morrison1986,
   author = {Philip J. Morrison},
   doi = {10.1016/0167-2789(86)90209-5},
   issn = {01672789},
   issue = {1-3},
   journal = {Physica D: Nonlinear Phenomena},
   title = {A paradigm for joined {Hamiltonian} and dissipative systems},
   volume = {18},
   year = {1986},
}

@article{Morrison1984,
   author = {Philip J. Morrison},
   doi = {10.1016/0375-9601(84)90635-2},
   issn = {03759601},
   issue = {8},
   journal = {Physics Letters A},
   title = {Bracket Formulation for Irreversible Classical Fields},
   volume = {100},
   year = {1984},
}

@article{Morrison1980,
	author = {Philip J. Morrison},
	doi = {https://doi.org/10.1016/0375-9601(80)90776-8},
	issn = {0375-9601},
	journal = {Physics Letters A},
	number = {5},
	pages = {383-386},
	title = {The {Maxwell-Vlasov} equations as a continuous {Hamiltonian} system},
	url = {https://www.sciencedirect.com/science/article/pii/0375960180907768},
	volume = {80},
	year = {1980}
}

@techreport{Morrison1981,
  author       = {Morrison, P J},
  title        = {Hamiltonian field description of two-dimensional vortex fluids and guiding center plasmas},
  institution  = {Princeton Plasma Physics Lab. (PPPL), Princeton, NJ (United States)},
  doi          = {10.2172/6351319},
  url          = {https://www.osti.gov/biblio/6351319},
  place        = {United States},
  year         = {1981},
  month        = {03}
}

@article{Morrison1982,
    author = {Morrison, Philip J.},
    title = "{Poisson brackets for fluids and plasmas}",
    journal = {AIP Conference Proceedings},
    volume = {88},
    number = {1},
    pages = {13-46},
    year = {1982},
    month = {07},
    issn = {0094-243X},
    doi = {10.1063/1.33633},
    url = {https://doi.org/10.1063/1.33633},
    eprint = {https://pubs.aip.org/aip/acp/article-pdf/88/1/13/11912899/13\_1\_online.pdf},
}

@article{Morrison2017,
    author = {Morrison, P. J.},
    title = {Structure and structure-preserving algorithms for plasma physics},
    journal = {Physics of Plasmas},
    volume = {24},
    number = {5},
    pages = {055502},
    year = {2017},
    month = {04},
    issn = {1070-664X},
    doi = {10.1063/1.4982054},
    url = {https://doi.org/10.1063/1.4982054},
    eprint = {https://pubs.aip.org/aip/pop/article-pdf/doi/10.1063/1.4982054/16678126/055502\_1\_online.pdf},
}

@article{Nanbu1997,
   author = {K. Nanbu},
   doi = {10.1103/PhysRevE.55.4642},
   issn = {1063651X},
   issue = {4},
   journal = {Physical Review E - Statistical Physics, Plasmas, Fluids, and Related Interdisciplinary Topics},
   title = {Theory of cumulative small-angle collisions in plasmas},
   volume = {55},
   year = {1997},
}

@article{Nocedal1980,
 ISSN = {00255718, 10886842},
 URL = {http://www.jstor.org/stable/2006193},
 author = {Jorge Nocedal},
 journal = {Mathematics of Computation},
 number = {151},
 pages = {773--782},
 publisher = {American Mathematical Society},
 title = {Updating Quasi-{Newton} Matrices with Limited Storage},
 urldate = {2024-10-29},
 volume = {35},
 year = {1980}
}

@article{Norton2014,
   author = {Richard A. Norton and G. R.W. Quispel},
   doi = {10.3934/dcds.2014.34.1147},
   issn = {10780947},
   issue = {3},
   journal = {Discrete and Continuous Dynamical Systems- Series A},
   title = {Discrete gradient methods for preserving a first integral of an ordinary differential equation},
   volume = {34},
   year = {2014},
}

@article{Perin2014,
	author = {M. Perin and C. Chandre and P.J. Morrison and E. Tassi},
	doi = {https://doi.org/10.1016/j.aop.2014.05.011},
	issn = {0003-4916},
	journal = {Annals of Physics},
	pages = {50-63},
	title = {Higher-order Hamiltonian fluid reduction of Vlasov equation},
	url = {https://www.sciencedirect.com/science/article/pii/S0003491614001250},
	volume = {348},
	year = {2014}
}

@misc{petsc-web-page,
  author = {Satish Balay and Shrirang Abhyankar and Mark~F. Adams and Steven Benson and Jed Brown
    and Peter Brune and Kris Buschelman and Emil~M. Constantinescu and Lisandro Dalcin and Alp Dener
    and Victor Eijkhout and Jacob Faibussowitsch and William~D. Gropp and V\'{a}clav Hapla and Tobin Isaac and Pierre Jolivet
    and Dmitry Karpeev and Dinesh Kaushik and Matthew~G. Knepley and Fande Kong and Scott Kruger
    and Dave~A. May and Lois Curfman McInnes and Richard Tran Mills and Lawrence Mitchell and Todd Munson
    and Jose~E. Roman and Karl Rupp and Patrick Sanan and Jason Sarich and Barry~F. Smith
    and Stefano Zampini and Hong Zhang and Hong Zhang and Junchao Zhang},
  title        = {{PETS}c {W}eb page},
  url          = {https://petsc.org/},
  howpublished = {\url{https://petsc.org/}},
  year         = {2022},
}

@manual{PETScManual,
   author = {Satish Balay and et. al.},
   title = {PETSc/TAO Users Manual},
   year = {2022},
   version = {Revision 3.18},
   organization = {Argonne National Laboratory},
   pagetotal = {310}
}

@article{Pusztay2022,
   author = {Pusztay, Joseph V. and Knepley, Matthew G. and Adams, Mark F.},
   title = {Conservative Projection Between Finite Element and Particle Bases},
   journal = {SIAM Journal on Scientific Computing},
   volume = {44},
   number = {4},
   pages = {C310-C319},
   year = {2022},
   doi = {10.1137/21M1454079},
   URL = {https://doi.org/10.1137/21M1454079},
   eprint = {https://doi.org/10.1137/21M1454079}
}

@article{Qin2016,
   author = {Hong Qin and Jian Liu and Jianyuan Xiao and Ruili Zhang and Yang He and Yulei Wang and Yajuan Sun and Joshua W. Burby and Leland Ellison and Yao Zhou},
   doi = {10.1088/0029-5515/56/1/014001},
   issn = {17414326},
   issue = {1},
   journal = {Nuclear Fusion},
   title = {Canonical symplectic particle-in-cell method for long-term large-scale simulations of the {Vlasov-Maxwell} equations},
   volume = {56},
   year = {2016},
}

@article{Quispel1996,
   author = {G. R.W. Quispel and G. S. Turner},
   doi = {10.1088/0305-4470/29/13/006},
   issn = {03054470},
   issue = {13},
   journal = {Journal of Physics A: Mathematical and General},
   title = {Discrete gradient methods for solving {ODEs} numerically while preserving a first integral},
   volume = {29},
   year = {1996},
}

@article{Shadwick2014,
   author = {B. A. Shadwick and A. B. Stamm and E. G. Evstatiev},
   doi = {10.1063/1.4874338},
   issn = {10897674},
   issue = {5},
   journal = {Physics of Plasmas},
   title = {Variational formulation of macro-particle plasma simulation algorithms},
   volume = {21},
   year = {2014},
}

@article{Squire2012,
   author = {J. Squire and H. Qin and W. M. Tang},
   doi = {10.1063/1.4742985},
   issn = {1070664X},
   issue = {8},
   journal = {Physics of Plasmas},
   title = {Geometric integration of the {Vlasov-Maxwell} system with a variational particle-in-cell scheme},
   volume = {19},
   year = {2012},
}

@article{Stamm2014,
   author = {Alexander B. Stamm and Bradley A. Shadwick and Evstati G. Evstatiev},
   doi = {10.1109/TPS.2014.2320461},
   issn = {00933813},
   issue = {6},
   journal = {IEEE Transactions on Plasma Science},
   title = {Variational formulation of macroparticle models for electromagnetic plasma simulations},
   volume = {42},
   year = {2014},
}

@Article{Vlasov1938,
 Author = {A. {Vlasov}},
 Title = {{\"Uber die Schwingungseigenschaften des Elektronengases}},
 FJournal = {{Zhurnal \`Eksperimentalno\u{\i} i Teoretichesko\u{\i} Fiziki}},
 Journal = {{Zh. \`Eksper. Teor. Fiz.}},
 ISSN = {0044-4510},
 Volume = {8},
 Pages = {291--318},
 Year = {1938},
 Publisher = {Academy of Sciences of the Union of Soviet Socialist Republics - USSR (Akademiya Nauk SSSR), Moscow; MAIK ``Nauka/ Interperiodika'', Moscow},
 Zbl = {0022.18102}
}

@article{Werner2025,
	author = {Werner, Gregory and Adams, Luke and Cary, John},
	doi = {10.48550/arXiv.2503.05123},
	month = {03},
	title = {Suppressing grid instability and noise in particle-in-cell simulation by smoothing},
	year = {2025}
}

@article{Xiao2018,
   author = {Jianyuan Xiao and Hong Qin and Jian Liu},
   doi = {10.1088/2058-6272/aac3d1},
   issn = {20586272},
   issue = {11},
   journal = {Plasma Science and Technology},
   title = {Structure-preserving geometric particle-in-cell methods for {Vlasov-Maxwell} systems},
   volume = {20},
   year = {2018},
}

@article{Xiao2015,
   author = {Jianyuan Xiao and Hong Qin and Jian Liu and Yang He and Ruili Zhang and Yajuan Sun},
   doi = {10.1063/1.4935904},
   issn = {10897674},
   issue = {11},
   journal = {Physics of Plasmas},
   title = {Explicit high-order non-canonical symplectic particle-in-cell algorithms for {Vlasov-Maxwell} systems},
   volume = {22},
   year = {2015},
}

@article{Xiao2016,
   author = {Jianyuan Xiao and Hong Qin and Philip J. Morrison and Jian Liu and Zhi Yu and Ruili Zhang and Yang He},
   doi = {10.1063/1.4967276},
   issn = {10897674},
   issue = {11},
   journal = {Physics of Plasmas},
   title = {Explicit high-order noncanonical symplectic algorithms for ideal two-fluid systems},
   volume = {23},
   year = {2016},
}

@article{Zonta2022,
    author = {Zonta, Filippo and Pusztay, Joseph V. and Hirvijoki, Eero},
    title = {Multispecies structure-preserving particle discretization of the {Landau} collision operator},
    journal = {Physics of Plasmas},
    volume = {29},
    number = {12},
    pages = {123906},
    year = {2022},
    month = {12},
    issn = {1070-664X},
    doi = {10.1063/5.0105182},
    url = {https://doi.org/10.1063/5.0105182},
    eprint = {https://pubs.aip.org/aip/pop/article-pdf/doi/10.1063/5.0105182/16633654/123906\_1\_online.pdf},
}

@misc{pusztay2023landaucollisionintegralparticle,
      title={The {Landau} Collision Integral in the Particle Basis in the {PETSc} Library}, 
      author={Joseph Pusztay and Filippo Zonta and Matt Knepley and Mark Adams},
      year={2023},
      eprint={2306.12606},
      archivePrefix={arXiv},
      primaryClass={physics.plasm-ph},
      url={https://arxiv.org/abs/2306.12606}, 
}

@manual{nrlformulary2023,
  title        = {NRL Plasma Formulary},
  author       = {J. D. Huba},
  organization = {Naval Research Laboratory},
  address      = {Washington, DC},
  year         = {2023},
  note         = {\url{https://www.nrl.navy.mil/Our-Work/Areas-of-Research/Plasma-Physics/Formulary/}},
}
\end{document}